\documentclass[review]{siamart1116}

\nolinenumbers

\usepackage{lipsum}
\usepackage{amsfonts}
\usepackage{graphicx}
\usepackage{subcaption}
\usepackage{epstopdf}
\usepackage{algorithm,algpseudocode}
\usepackage{tikz-cd}
\usepackage{tikz}
\usetikzlibrary{patterns}
\ifpdf
  \DeclareGraphicsExtensions{.eps,.pdf,.png,.jpg}
\else
  \DeclareGraphicsExtensions{.eps}
\fi

\numberwithin{theorem}{section}

\newcommand{\TheTitle}{Computing quasiconformal folds} 
\newcommand{\TheAuthors}{Di Qiu, Ka-Chun Lam, Lok-Ming Lui}
\newsiamthm{problem}{Problem}
\newsiamthm{remark}{Remark}
\newsiamthm{assumption}{Assumption}
\headers{\TheTitle}{\TheAuthors}

\title{{\TheTitle}}

\author{
  Di Qiu\thanks{Department of Mathematics, The Chinese University of Hong Kong.
    \email{dqiu@math.cuhk.edu.hk}.}
    \and
  Ka-Chun Lam\thanks{Department of Computing + Mathematical science, California Institute of Technology. \email{kclam@caltech.edu}).} 
  \and
  Lok-Ming Lui\thanks{Department of Mathematics, The Chinese University of Hong Kong.
    (\email{lmlui@math.cuhk.edu.hk}).} }

\usepackage{amsopn}

\ifpdf
\hypersetup{
  pdftitle={\TheTitle},
  pdfauthor={\TheAuthors}
}
\fi

\begin{document}

\maketitle

\begin{abstract}
  Computing surface folding maps has numerous applications ranging from computer graphics to material design. In this work we propose a novel way of computing surface folding maps via solving a linear PDE. This framework is a generalization to the existing computational quasiconformal geometry and allows precise control of the geometry of folding. This property comes from a crucial quantity that occurs as the coefficient of the equation, namely the alternating Beltrami coefficient. This approach also enables us to solve an inverse problem of parametrizing the folded surface given only partial data with known folding topology. Various interesting applications such as fold sculpting on 3D models, study of Miura-ori patterns, and self-occlusion reasoning are demonstrated to show the effectiveness of our method.
\end{abstract}



\section{Introduction}
Modeling the folding phenomena of surfaces, as well as the study of its regular patterns such as mathematical origami, has attracted lots of interests in computer graphics, material design as well as mathematics. Recent examples include origamizing surfaces \cite{demaine2017origamizer}, material design with mathematical origami \cite{dudte2016programming}, a FoldSketch system manipulating the folding of clothes \cite{li2018foldsketch}, modeling curved folding surface used in fabrication and architectural design, see for example \cite{kilian2008curved, zhu2013soft}. 
All these approaches treat a folded surface as an embedded surface in the three dimensional space, with the folding being a property of the embedding, rather than something intrinsic of the surface itself.

Very often such an embedding can be non-smooth, and thus can be harder to deal with directly. For example, a typical origami when viewed as an isometric embedding of some flat surface, is only of $C^{0}$ regularity. In other words, the embedding is not differentiable at the folding creases, which however are the key trait of such folded surfaces. 

Is it possible to encode the folding as something intrinsic to the surface, in the sense of differential geometry? Since then it would allow us to abstract away some excessive degrees of freedom from all possible embeddings. Such an encoding will have to keep the key properties of the folding, so as to manipulate the folding via the encoding in effective manners. 

In this paper we propose to use a novel fold encoding scheme that is intrinsic to the surfaces, taking ideas from conformal geometry. We encode folded surface by a ``folding map" defined on its encoding domain.  The folding map can be determined using a local differential quantity called the {\it alternating Beltrami coefficient} \cite{srebro1996branched}, which takes the form of a complex-valued scalar field defined on the encoding domain. The coefficient represents precisely a more general type of local conformal distortion of the folding map, including its possible {\it change of orientation}, which we use to model folds and its geometry. Our work can be considered as the first application of quasiconformal methods to the problem of modeling and studying the folding of surfaces.

Computationally, the folding map is the solution of a linear partial differential equation, called {\it alternating Beltrami equation} in quasiconformal theory.  By solving this equation based on the associated coefficients, we can fold a domain into the corresponding folded surface. The inverse of the folding map can also be computed using the equation of the same type defined on the folded surface, which unfolds it back to its encoding domain.  We show that the alternating Beltrami equation admits a quadratic variational formulation, and in fact it generalizes the classical least square conformal parametrization \cite{levy2002least, desbrun2002intrinsic} in the sense that it allows folding with prescribed conformal distortion. Consequently, folding or unfolding can be achieved by simply solving a sparse linear system.  

Moreover, the encoding scheme with its associated folding or unfolding maps also allows us to relate different folded surfaces with the same  ``folding topology" via deformation.
With this we can formulate optimization problems to find the ``best" positioning of the folding creases under certain criteria by deforming an existing one.
We use this idea to solve the problem of recovering a {\it flat-foldable} surface with self-occlusion. In other words, we show that it is possible in many cases to use only the topology and some partial geometric information of the folding to reconstruct the entire flat-foldable surface.

The paper is organized as follows. In Section \ref{subsec2.1} we develop the topological notion of ``folding homeomorphism" and the associated singular set configuration, which are formal definitions of the ``folding map" and ``folding topology" mentioned previously. We then explain the geometric meaning of the alternating Beltrami coefficient and the associated equation, visualize the folding effect, and give the definition of flat-foldability. In Section \ref{subsec2.2} we develop the variational formulation and give its discrete geometric interpretation, together with its discretization and implementation. In Section \ref{sec3} we define the reconstruction problem of flat-foldable surfaces with self-occlusion, and develop the ``reinforcement iteration" algorithm to solve it. Finally, in Section \ref{sec4} and \ref{sec5} we demonstrate the results of the algorithm and various interesting applications including generating new flat-foldable {\it Miura-ori} prototypes, fold-like texture generation, fold and cusp sculpting on surfaces, fold in-painting for surfaces, and occlusion reasoning for flat-foldable surfaces.

\subsection{Related work} Here we briefly list some important related works in this area, while they are by no means an exhaustive survey. 

{\bf Computational quasiconformal geometry.} Computational technique of conformal map \cite{levy2002least, desbrun2002intrinsic} turned out to be very useful in computer graphics. Since the seminal work of Gu and Yau \cite{gu2004genus}, the conformal geometry framework in surface processing tasks has advanced significantly. Earlier work generalizing these ideas is already implicit in the work of Seidel \cite{zayer2005discrete}. The quasiconformal extension of the surface registration framework was proposed by Lui and his coauthors \cite{lui2012optimization, lui2014teichmuller, choi2015flash}, with successful applications to medical image registration and surface registration. The quasiconformal method is able to handle large deformations, where conformal methods typically fail. Our work is a generalization along this line, which allows the manipulation of folding, and opens up a new research area to explore.

{\bf Modeling surface folding and mathematical origamis.} In computer graphics there has been a notable amount of work on modeling the folding phenomena of surfaces. Many interesting works focus on 3D interactive design. These include the method of thin plate form with explicit user control of folding angles for interactive 3D graphics design in \cite{zhu2013soft}, which can also achieve the sharp folding edges as we do here. Our framework and techniques are completely different, especially here we are taking advantage of the fact that alternating Beltrami equation can be solved effectively in 2D. On the other hand, there are studies of developable surface design with curved folding \cite{kilian2008curved}, taking the advantage of the special quad meshes, while we don't have this restriction, but the focus and techniques are rather different. We must also mention the work of Demaine and Tachi \cite{demaine2017origamizer}, who developed algorithms to fold a planar paper into arbitrary 3D shapes. The study of the folding phenomena also has industrial applications in such as the 4D printing of \cite{kwok2015four} and material design in \cite{dudte2016programming}. In comparison with the work above, we encode the folding and its geometry in an intrinsic way, and hence three dimensional geometric features of the folded surface, such as mean curvature, do not belong to our framework but can be tackled by pre- or post-processing. We expect to discover interesting connections with others in future work.

\section{Computing quasiconformal folds} \label{sec2}
\subsection{Definitions of folding homeomorphism, singular set configuration and derivation of alternating Beltrami equation} \label{subsec2.1}
\subsubsection{Topological description}
The notion of folding used in this paper departs from those of three
dimensional nature, as in \cite{kilian2008curved}. Instead, we model it via the the
continuous map from the domain surface to the target surface with
designated change of orientation. This is made precise in the following
definitions.
\begin{definition}[Two color map]
Let $X$ be a two dimensional surface and $\Sigma$ be a union of finite number of curves in $X$. A pair $(X, \Sigma)$ is said to be a two color map if $X\setminus \Sigma$ is a disjoint union of surfaces whose boundary belongs to $\Sigma$, and each member of the union is assigned uniquely with one of the two distinct colors such that the closures of any two distinct members assigned with the same color intersect at most at finite number of points.
\end{definition}
\begin{definition}[Folding homeomorphism and its singular set configuration]
Let $X$ and $Y$ be oriented surfaces. A continuous and {\it discrete} map $f:X\to Y$ 
($f^{-1}(y)$ are isolated in $X$ for all $y\in Y$) is called a
{\it folding homeomorphism} if there is a subset $\Sigma\subset X$, of Hausdorff dimension $1$,
with $(X,\Sigma)$ forming a white-balck two color map, such that when restricting to
the the white (or black) region, $f$ is an orientation-preserving (or -reversing)
homeomorphism. The two color map $(X,\Sigma)$ is called the {\it singular set configuration} of $f$, and is sometimes simply referred to as $\Sigma $ if it is clear from the context. The white and black regions are denoted by $X^+, X^-$ respectively. \end{definition} 

Note that we have required the map to be discrete in order to avoid the case of degeneracy. 

For convenience we would also like to give names to the points in the singular set
according to the properties of $f$. The definition below will classify the usually encountered situation and suffices for our purposes in this paper. Readers can also refer to Gutlyanskii et al. \cite{gutlyanskii2012alternating} for related materials. 
\begin{definition}[Folding point and cusp point] \label{def:1}
Let $f:X\to Y$ be a folding homeomorphism with singular set configuration $\Sigma$. 
\begin{itemize} 
\item A point
$p\in\Sigma$ is called a folding point if there is an open neighborhood $U$ of $p$ such that
$U\setminus\Sigma$ is disconnected into exactly $2$ simply connected components,
and the restriction $f\big|_{U}$ is topologically equivalent to $(x,y)\mapsto(x,|y|)$, where
$U\cap\Sigma$ plays the role of $x$-axis. 
\item A point
$p\in\Sigma$ is called a cusp point if there is an open neighborhood $U$ of $p$ such that
$U\setminus\Sigma$ is disconnected into exactly $2n$ simply connected components, $n>1$, and the remaining
points $p'\in (\Sigma\cap U)\setminus\{p\}$ are all folding points.
\item The collection of paths in $\Sigma$ consisting of all folding points is called folding lines.
\end{itemize}
\end{definition} 

We illustrate these concepts in Figure \ref{fig:def_fold}. 
For a simple explicit example, consider the map
\[
f(x+iy) = \begin{cases}
x+iy & \text{ if $ax + by \geq 0$} \\
x+iy - \frac{2(ax+by)(a + bi)}{a^2 + b^2} & \text{ if $ax + by  < 0$} 
\end{cases}
\]
where $a>0, b\in\mathbb{R}$. Then we see that $f$ folds the half plane $\{x+iy: ax + by  < 0\}$ into the other half $\{x+iy: ax + by  > 0\}$. Explicit examples of cusp is more involved to write down, but can be found in \cite{gutlyanskii2012alternating}. 
A more famous example is the paper crane origami, whose singular set configuration is shown in Figure \ref{fig:sing set}, where for better visualization we use yellow and purple instead of white and black. The paper crane origami is in fact {\it flat-foldable}, which we will define in Definition \ref{def: qc_fold}. \\
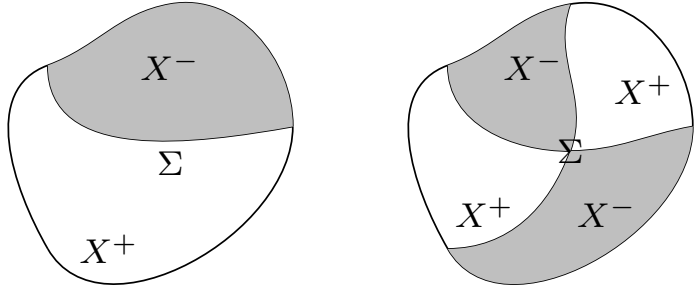
\begin{figure}
 \centering
 \begin{subfigure}[b]{0.35\textwidth}
 \resizebox {\columnwidth} {!} {
\begin{tikzpicture}
    \draw (1,0) to [out=120,in=200]  (1,1.5);
    \draw (1,1.5) to [out=20,in=190] (2,2) to [out=10,in=90] (3,1);
    \draw (3,1) to [out=270,in=300] (1,0);
	\draw (1,1.5) to [out=-90,in=190] (3,1);
	\fill [color = lightgray] (1,1.5) to [out=-90,in=190] (3,1) to [out=90,in=10] (2,2) to [out=190,in=20] (1,1.5);
	\node at (1.5,0) {\footnotesize $X^{+}$};
	\node at (2,1.5) {\footnotesize $X^{-}$};
	\node at (2,0.7) {\footnotesize $\Sigma$};
  \end{tikzpicture} }
      \caption{The singular configuration contains only folding points: the folding line separates $X^{-}$ and $X^{+}$.}
    \label{fig:folding_point}
     \end{subfigure}
\hspace{15pt}
\begin{subfigure}[b]{0.35\textwidth}
\resizebox {\columnwidth} {!} {
    \begin{tikzpicture}
    \draw (1,0) to [out=120,in=200]  (1,1.5);
    \draw (1,1.5) to [out=20,in=190] (2,2) to [out=10,in=90] (3,1);
    \draw (3,1) to [out=270,in=300] (1,0);
	\draw (1,1.5) to [out=-90,in=180] (2,0.8);
	\draw (2,0.8) to [out=0,in=190] (3,1);
	\draw (1,0) to [out=0,in=250] (2,0.8);
	\draw (2,0.8) to [out=70,in=250] (2,2);
	\fill [color = lightgray] (1,1.5) to [out=-90,in=180] (2,0.8) to [out=70,in=250] (2,2) to [out=190,in=20] (1,1.5);
	\fill [color = lightgray] (1,0) to [out=0,in=250] (2,0.8) to [out=0,in=190] (3,1) to [out=270,in=300] (1,0);
	\node at (1.3,0.3) {\footnotesize $X^{+}$};
	\node at (1.7,1.5) {\footnotesize $X^{-}$};
	\node at (2.6,1.3) {\footnotesize $X^{+}$};
	\node at (2.3,0.3) {\footnotesize $X^{-}$};
	\node at (2,0.8) {\footnotesize $\Sigma$};
  \end{tikzpicture} }
  \caption{The singular configuration contains both folding points and a single cusp point: the cusp point joins the folding lines.}
  \end{subfigure}
  \caption{Illustration of Definition \ref{def:1}}
    \label{fig:def_fold}
\end{figure}

\begin{figure}
   \centering
    \begin{subfigure}[b]{0.38\textwidth}
        \includegraphics[width=\textwidth]{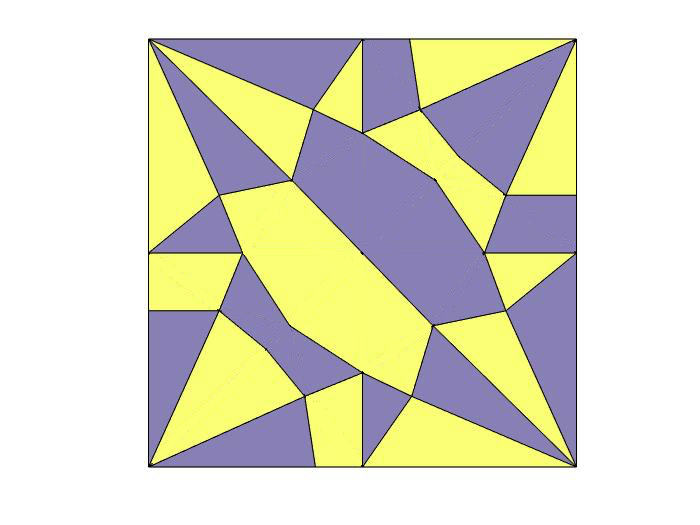}
        \caption{Singular set configuration of the paper crane. }
        \label{fig:crane_foldpattern}
    \end{subfigure} 
    \hspace{15pt}
    \begin{subfigure}[b]{0.38\textwidth}
        \includegraphics[width=\textwidth]{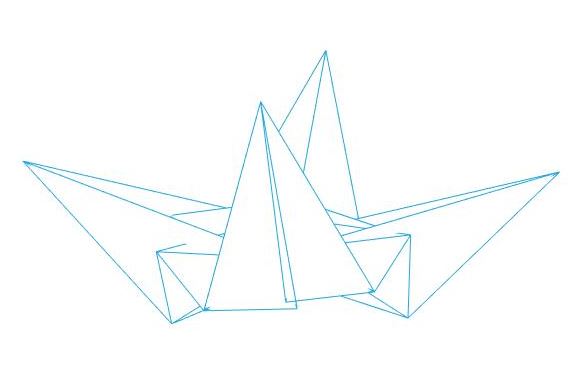}
        \caption{The paper crane obtained by solving an alternating Beltrami equation.}
        \label{fig:crane}
    \end{subfigure}
    \caption{The paper crane origami}
    \label{fig:sing set}
\end{figure}

Since change of orientation necessarily creates non-bijectivity, the inverse of a folding homeomorphism $f$ is no longer a function. Nevertheless it is still possible to define an ``unfolding homeomorphism" that unfolds a folded surface once we can distinguish the image of the points in $f^{-1}(p)$. Formally, one can consider the map
\[
F: (p , p) \mapsto (p , f(p)), \; p\in\Omega
\]
where $f$ is a folding homeomorphism. $F$ is bijective, hence has an inverse that unfolds the folded surface. In practice, if the fold is represented as a mesh as in Figure \ref{fig:crane}, the points in the image of the folding map are naturally distinguished by the mesh data structure. Locally, the unfolding homeomorphism undoes the change of orientation done by the corresponding folding homeomorphism. As we shall see below in Remark \ref{rem:unfold}, since we are dealing with differential equation, we need only to know on which triangle the orientation is reversed, and thus need not worry about the non-bijectivity. 

\subsubsection{The quasiconformal geometry of folding homeomorphisms} So far the above definitions have been topological. We can equip the surfaces with
Riemannian metrics and thus talk about the intrinsic geometry of the
folding homeomorphism $f:(X,g_{1})\to(Y,g_{2})$ between Riemannian
surfaces. 
To measure the local distortion of $f$ is to compare the pullback
metric $f^{*}g_{2}$ with the original metric $g_{1}$. 
In the other
direction, the knowledge of the distortion will give rise to a PDE
that characterizes the map. We shall describe this point in
details below. 

For our purpose, let us take an open set $\Omega \subset\mathbb{R}^2 \cong \mathbb{C}$
in $X$ which does not contain points in $\Sigma$, and put $g_{2}$ to be the Euclidean metric. If we consider
the pull back metric by the unknown map $f$ as given data in the form of a matrix field $H:\Omega\to\mathbf{S}_{++}$,
where $\mathbf{S}_{++}$ denotes the space of symmetric positive definite
$2\times2$ matrices, and assume $f$ is differentiable, it then satisfies
the nonlinear system \cite{astala2008elliptic}.
\[
Df(z)^{T}Df(z)=H(z),\ \forall z\in \Omega,
\]
where $f = (u,v)^T, z = (x,y)^T, Df(z) = \begin{bmatrix}u_{x} & u_{y}\\
v_{x} & v_{y}
\end{bmatrix}$.
Surprisingly enough, it is possible to reduce the above to a linear
equation if the data is given up to multiplying a everywhere positive
function. This is the essential advantage for us to introduce conformal
geometry in dimension two in our problem. To do this, denote 
\[
\mathbf{S}(2)=\{M\in\mathbf{S}_{++}:\det M=1\}.
\]
And let $Q:\Omega\to S(2)$ be such that $H(z) = \phi(z)Q(z)$, so the above nonlinear system is expressed
as
\[
Df(z)^{T}Df(z)=\phi(z)Q(z),\ \phi(z)>0.
\]
By taking determinants on both sides, we get
\[
\phi(z)=|J_{f}(z)|
\]
where $J_{f}(z)=\det Df(z)$. Note that the absolute value is necessary
since $f$ may be orientation reversing. As a result, we get
\[
Df(z)^{T}Df(z)=|J_{f}(z)|Q(z).
\]
Multiplying $Df(z)^{-1}$ on the right of both sides, and write $f=(u,v)^{T}$, \hspace{1pt}
$Q=\begin{bmatrix}q_{11} & q_{12}\\
q_{12} & q_{22}
\end{bmatrix},$ we obtain the system
\[
\begin{bmatrix}u_{x} & u_{y}\\
v_{x} & v_{y}
\end{bmatrix}^{T}=\text{sgn}(J_{f}(x))\cdot\begin{bmatrix}q_{11} & q_{12}\\
q_{12} & q_{22}
\end{bmatrix}\begin{bmatrix}v_{y} & -u_{y}\\
-v_{x} & u_{x}
\end{bmatrix}.
\]
It is a straightforward matter to rewrite the above system in complex
derivatives, obtaining
\[
\frac{\partial f}{\partial\bar{z}}(z)=\mu(z)\frac{\partial f}{\partial z}(z),
\]
where $\mu=\frac{q_{11}-q_{22}+2iq_{12}}{q_{11}+q_{22}+2\cdot \text{sgn}(J_f(x)}, \frac{\partial f}{\partial\bar{z}}=(u_{x}-v_{y})+i(u_{y}+v_{x})$
and $\frac{\partial f}{\partial z}=(u_{x}+v_{y})+i(-u_{y}+v_{x})$. 
This is called
the {\it alternating Beltrami equation}, coined by Srebro and Yakubov \cite{srebro1996branched}. The name refers to the fact that
\[
\begin{cases}
|\mu|<1 & \text{if }\text{sgn}(J_{f}(x))>0\\
|\mu|>1 & \text{if }\text{sgn}(J_{f}(x))<0
\end{cases},
\]
and hence it differs from the classical Beltrami equation in that the modulus of the coefficient is only required to be bounded away from $1$.
This motivates a more analytic definition for the folding homeomorphism, which is the principal mathematical subject of this paper.

\begin{definition} [Quasiconformal map with folds] \label{def: qc_fold}
A folding homeomorphism $f:\Omega\subset\mathbb{C}\to\mathbb{C}$, $K\geq1$ is called a  generalized $K$-quasiconformal map with singular configuration $\Sigma$
if it is a solution of a
alternating Beltrami equation $\frac{\partial f}{\partial\bar{z}}(z)=\mu(z)\frac{\partial f}{\partial z}(z)$, such that
in $\Omega^{+}$ and $\Omega^{-}$ it holds that $|\mu(z)|<1$ and $|\mu(z)|>1$, respectively, and moreover satisfies the bound
\[
\big|\frac{1+|\mu(z)|}{1-|\mu(z)|}\big|\leq K
\]
for all $z\in\Omega$ except for a set of Lebesgue measure zero. In particular, the case of $K=1$ will be called generalized conformal, or flat-foldable, or a planar origami.
\end{definition}
The quotient inside the bound has the interpretation of linear distortion of the map, see \cite{astala2008elliptic}. The above definition of flat-foldability is adapted to our problem, in particular in a discrete, computational setting. It includes the case where a surface is rigidly flat-folded, whose folding lines of the singular set configuration are all Euclidean geodesics. 

\begin{remark}
The rigidity associated to flat-foldability is also manifested via a well-known condition about how folding lines join each other at a cusp point, known as the {\it Kowasaki's condition}. In details, let $n>1$ be an integer and suppose there are $2n$ Euclidean geodesics emanating from a cusp point $p\in U\subset\mathbb{C}$.
Then the neighborhood $U$ is flat-foldable if the alternating sum of the angles $(\alpha_i)_{i=1}^{2n}$ formed by every 
two neighboring Euclidean geodesics
satisfy the  condition
$$
\sum_{i=1}^{2n}(-1)^{i}\alpha_{i} = 0. 
$$
This condition is utilized in \cite{dudte2016programming} for constrained optimization.
However, this formalism will not play a significant role in the algorithms we propose in this paper. 
\end{remark}

To get a better picture of the alternating Beltrami equations, 
we illustrate it with the effect of the Beltrami coefficients on a
single triangle (i.e. the linearized effect at the tangent space level). 
This should provide one with geometric intuition for the
solutions on a triangulated mesh.

Let us rewrite the Beltrami equation as a system of first-order
PDEs in the usual Cartesian coordinate. Suppose $f:(x,y)\mapsto(u,v)$
satisfies the equation $\frac{\partial f}{\partial\bar{z}}(z)=\mu(z)\frac{\partial f}{\partial z}(z)$.
If we write $\mu=\rho+i\tau$,
then it's not hard to see that 

\begin{equation}\label{eq:1} 
\begin{bmatrix}u_{y}\\
v_{y}
\end{bmatrix}=\frac{1}{(1+\rho){}^{2}+\tau^{2}}\begin{bmatrix}2\tau & |\mu|^2-1\\
1-|\mu|^2 & 2\tau
\end{bmatrix}\begin{bmatrix}u_{x}\\
v_{x} 
\end{bmatrix},  \end{equation} 
here we have assumed $\rho\neq-1$ and $\tau\neq0$.

Hence for a single triangle, on which we assume $f$ is linear, the
map is determined up to a similarity transform (uniform scaling and rotation) in the target
domain. Now suppose $f$ maps a domain triangle $[v_{1},v_{2},v_{3}]=[(0,0),(1,0),(x,y)]$
to the target triangle $[w_{1},w_{2},w_{3}]=[(0,0),(1,0),(u(x,y),v(x,y))]$.
Then
\begin{equation}
\begin{bmatrix}u(x,y)\\
v(x,y)
\end{bmatrix}=\begin{bmatrix}1 & \frac{2\tau}{(1+\rho){}^{2}+\tau^{2}}\\
0 & \frac{1-|\mu|^2}{(1+\rho){}^{2}+\tau^{2}}
\end{bmatrix}\begin{bmatrix}x\\
y
\end{bmatrix}.\label{eq:2}
\end{equation}
One can check that the set of points for the family of $\mu$ with
each fixed modulus $|\mu|\neq1$ form a circle, whereas in the case
$|\mu|=1$ the circle degenerates to the $x$-axis.  An illustration of this fact is shown in Figure \ref{fig:mu_plot}.

As one goes beyond to the case $|\mu|>1$, where the anti-diagonal
terms experience a change of sign, this leads to a ``flipping'' of
the triangle. In fact, for a single triangle, everything remains the same after a mirror reflection about the $x$-axis,
and the case $\mu=\infty$ corresponds to the {\it anti-conformality} of the map. Here $\infty$ should be understood as the infinity point in the Riemann sphere. What is more,
for each fixed argument $\arg(\mu)$ and let the modulus $|\mu|$ vary, the set of solution points form an arc of
a circle, passing through the points $(x,y)$ and $(x,-y)$. Altogether, we see that the Beltrami coefficients in effect form a
bipolar coordinate in the plane containing the target triangle. Therefore,
it describes all possible angular distortion at the tangent space
level, including those having a change of orientation.\\
\begin{figure}
    \centering
    \includegraphics[width=\textwidth]{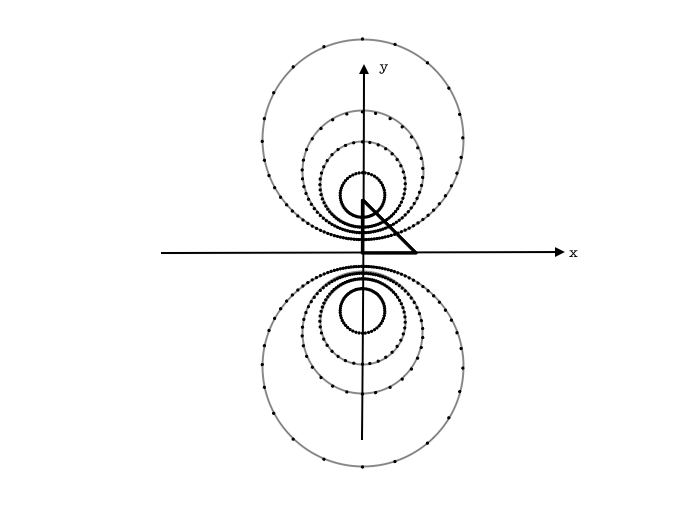}
    \caption{The trajectory of the third vertex under different values of Beltrami coefficients. Circles represent the situation when $|\mu| = 1/5, 7/20, 9/20, 3/5, 5/3, 20/9, 20/7, 5$ respectively.}
    \label{fig:mu_plot}
\end{figure}

\subsection{Energy formulation} \label{subsec2.2}
In this section we turn to the computational methods to solve the Beltrami equation. Previously proposed methods include the {\it Beltrami holomorphic flow method} as in \cite{lui2012optimization}, or the decoupling method as in \cite{lui2014teichmuller}. Both require entire boundary information for solving the Beltrami equation and it can be unrealistic in applications. Fortunately, it turns out to be also unnecessary once we realize the coupling of the two coordinate functions of the map. This coupling arises naturally in an energy functional of least squares type. For completeness, we analyze this problem below since we did not find it in the literature.

The formulation here takes inspiration from the well-known {\it least square conformal energy}, studied in \cite{levy2002least, desbrun2002intrinsic}, which take into account the coupling of $u$ and $v$. 
Its continuous formulation is 
 $$\int_{\Omega}\|\nabla u +\begin{bmatrix}0 & -1\\
 1 & 0
 \end{bmatrix} \nabla v\|^{2}dxdy.$$
The corresponding matrix associated to its discretization 
is the well-known {\it cotangent weight} matrix minus
a certain ``area matrix" \cite{pinkall1993computing, desbrun2002intrinsic}. This area matrix in fact plays the role of certain Neumann boundary condition. 
One would expect analogous results to hold in the quasiconformal setting. 

But in the quasiconformal case, it is not an entirely trivial matter to formulate the correct analog. It turns out the energy must give rise to a pair of second order elliptic equations as a necessary condition.

\subsubsection{The decoupling method and the necessary condition} Perhaps the most straightforward way to
solve the Beltrami equation is to decouple the corresponding
first order system into two independent second order equations, namely,
\begin{proposition}[Necessary condition]\label{prop:1}  
Suppose for $z\in \Omega \setminus \Sigma$, $f(z) = u(z)+iv(z)$ satisfies the equation 
$\frac{\partial f}{\partial\bar{z}}(z)=\mu(z) \frac{\partial f}{\partial z}(z)$. Assume the domain $\Omega$ is given the usual Euclidean geometry, and $|\mu|\neq1$, 
then we have
\begin{equation}
\begin{cases}
-\nabla\cdot(A\nabla u(z)) & =0\\
-\nabla\cdot(A\nabla v(z)) & =0
\end{cases}\label{eq:3}
\end{equation}
where $A= \frac{1}{1-|\mu|^2}\begin{bmatrix}(\rho-1)^{2}+\tau^{2} & -2\tau\\
-2\tau & (1+\rho)^{2}+\tau^{2}
\end{bmatrix}$,  and $\mu = \rho+i\tau$. 
 \end{proposition} 
\begin{proof}
Observe that $\frac{\partial f}{\partial\bar{z}}(z)=\mu(z) \frac{\partial f}{\partial z}(z)$ can be transformed into
$$ 
\begin{bmatrix} u_x \\u_y \end{bmatrix}
= \begin{bmatrix} 0 & 1 \\-1 & 0 \end{bmatrix}A\begin{bmatrix} v_x \\v_y \end{bmatrix}.
$$
Then making use of the commutativity of second order partial derivatives $u_{xy} = u_{yx}$ under the Euclidean coordinate,
we obtain $$\nabla\cdot(A\nabla u(z)) =0.$$ The other equation is obtained in a similar way.
\end{proof}
\begin{remark} \label{rem: alter}
Note that the coefficient matrix $A$ is positive (or negative) definite if $|\mu|< 1$ (or $|\mu|>1$, respectively). If $U$ is any open neighborhood, on which $A$ is either positive or negative but not both, then it is not hard to see that system \eqref{eq:3} are the Euler-Lagrange equations of the  
Dirichlet type energies
\begin{equation}
E_{\tilde{A}}(u;U)=\frac{1}{2}\int_{U}\|\tilde{A}^{1/2}\nabla u\|^{2}dxdy,\quad E_{\tilde{A}}(v;U)=\frac{1}{2}\int_{U}\|\tilde{A}^{1/2}\nabla v\|^{2}dxdy,\label{eq:8}
\end{equation}
with Dirichlet boundary conditions, where $\tilde{A}$ denotes $A$ if $A$ is positive definite, or $-A$ if $A$ is negative definite. 
Therefore, we see that in general the global variational problem
must be separated according to whether $|\mu|<1$ or $|\mu|>1$ in the domain $\Omega$. We shall often denote $\Omega^{+}$ (or $\Omega^{-}$)
to be the largest open subset such that $|\mu|<1$ (or $|\mu|>1$, respectively), which is consistent with the previous notation.
\end{remark}

The derived system \eqref{eq:3}
is a {\bf necessary condition} that in principle should be satisfied by any other method which solves the equation in the Euclidean domain $\Omega$. This motivates the following.

\begin{definition} \label{def:qc energy}
The {\bf least squares quasiconformal energy} of the map $z=(x,y)\mapsto (u,v)$ against Beltrami coefficient $\mu=\rho + i\tau$  is defined to be
\[
\begin{aligned}
E_{LSQC}(u,v;\mu)&=\frac{1}{2}\int_{\Omega}\|P\nabla u+ JP\nabla v\|^{2}\,dxdy 
\end{aligned}
\]
 where
\begin{equation*}
P = \frac{1}{\sqrt{1-|\mu|^2}}\begin{bmatrix}1-\rho & -\tau \\ -\tau & 1+\rho\end{bmatrix}, \; 
J =\begin{bmatrix}0 & -1\\
 1 & 0
 \end{bmatrix}
\end{equation*}
so that $P^TP = A$ as in \eqref{eq:3}. Note that if $|\mu| > 1$ we allow entries of $P$ to be imaginary, and $\|\cdot\|^2 = \langle \cdot, \cdot\rangle $ remains as the usual Euclidean norm (not Hermitian).
\end{definition}

Let's first consider the uniform case and assume that $|\mu|<1$. The case of $|\mu|>1$ defers by a negative sign. Thanks to the important observation that
$P^T J P = J$
, we have the following identity
$$
E_{LSQC}(u,v;\mu) = \left(E_{A}(u;\Omega)+E_{A}(v;\Omega)\right)-\mathcal{A}(u,v), 
$$
where $A$ is the same matrix described previously in \eqref{eq:3}, and 
$$\mathcal{A}(u,v):=\int_{\Omega}(u_{y}v_{x}-u_{x}v_{y})\,dxdy $$ 
is the area of the target surface. 

\begin{remark} \label{rem: lowbd_dirichlet}
Assume that $|\mu|<1$. Observe that we have obtained the analog of the classical lower bound of the Dirichlet energy
\begin{equation} \label{eq:13}
E_{A}(u)+E_{A}(v) \geq \mathcal{A}(u,v).
\end{equation}
This simply follows from the fact that $E_{LSQC}(u,v; \mu) \geq 0$. The vanishing of this energy is equivalent to the existence of $f = u+iv$ as a solution of the Beltrami equation with coefficient $\mu$. The existence is guaranteed for measurable Beltrami coefficients $\mu$ with $\|\mu\|_{L^{\infty}(\Omega)} < 1$, known as the {\it measurable Riemann map theorem}. Note also that the solution of the Beltrami equation is unique up to post-composition of conformal maps \cite{astala2008elliptic}.
\end{remark}
If we assume the domain $\partial\Omega$ has Lipschitz boundary, then the quantity $\mathcal{A}(u,v)$ 
is equal to the following integral on the boundary 
$$
\frac{1}{2}\int_{\partial\Omega}(v\nabla u-u\nabla v)\times\nu \,d\Gamma, $$
where $\nu(z)$ is the outer unit normal vector, and $d\Gamma$ is the standard measure of $\partial\Omega$. Actually, the coupling between $u$ and $v$ is realized as certain boundary condition applied to solving \eqref{eq:3}. The following  derivation of the second order equations with boundary condition is standard.
\begin{theorem}
Suppose $\mu$ is uniformly bounded away from $1$, $\Omega$ is connected with Lipschitz boundary, and suppose there exists one pair $(u,v)$, 
$u,v\in W^{2,2}(\Omega)$, such that 
\[ 
E_{LSQC}(u,v;\mu) = \arg\inf_{\tilde{u},\tilde{v}\in W^{1,2}} E_{LSQC}(\tilde{u},\tilde{v};\mu),
\]
then they satisfy the following Neumann boundary problem
\begin{equation} \label{eq: qc}
\begin{cases}
-\nabla\cdot(A\nabla u)  =0 \text{ in } \Omega \\
-\nabla\cdot(A\nabla v)  =0 \text{ in } \Omega \\
\partial_{A\nu}u +\nabla v \times \nu =0 \text{ on } \partial\Omega\\
\partial_{A\nu}v -\nabla u \times \nu =0 \text{ on } \partial\Omega
\end{cases},
\end{equation}
where as before $\nu(z)$ is the outer unit normal vector.
\end{theorem}

\subsubsection{Generalized quasiconformal energy}

 Because of the change of orientation, the energy formulation and the associated system of equations has to be accordingly modified. It is then crucial to study the interaction between regions of the domain that corresponds to different orientations of $f$. 

First of all, it follows from arguments in Remark \ref{rem: alter} and Remark \ref{rem: lowbd_dirichlet} that the alternating Beltrami equation, when {\it restricted to regions of constant orientation}, is equivalent to vanishing of the ``energies"
\[
\begin{aligned}
 E_{LSQC}^+(u,v;\mu) := \frac{1}{2}\int_{\Omega^+}\|P\nabla u+JP\nabla v\|^{2}\,dxdy = 0 \\
 E_{LSQC}^-(u,v;\mu) : = \frac{1}{2}\int_{\Omega^-}\|P\nabla u+JP\nabla v\|^{2}\,dxdy = 0
\end{aligned}
\]
where $\Omega^+ = \text{int } \{z\in \Omega: |\mu(z)| <1 \}$, $\Omega^- = \text{int } \{z\in \Omega: |\mu(z)| >1 \}$. Recall that when $|\mu|>1$  we allow entries of $P$ to be imaginary, hence $E_{LSQC}^-(u,v;\mu)$ is in fact negative definite. 

To obtain the global solution, one could solve the equation individually in $\Omega^+$ and $\Omega^-$ and glue the solution along the singular set configuration. It turns out that this can be done implicitly. The problem now is how to combine the quasiconformal energies on regions with different orientations into a single ``energy", so that we can solve the alternating Beltrami equation on the entire domain in one shot. 
\begin{theorem} [Generalized quasiconformal energy] \label{def:gqc energy}
Assume there are only finitely many cusp points. Define the generalized quasiconformal energy with Beltrami coefficient $\mu$ of the map $z=(x,y)\mapsto (u,v)$ in $W^{2,2}$ to be
\[
E_{GQC}(u,v;\mu) = E_{LSQC}^+(u,v;\mu) - E_{LSQC}^-(u,v;\mu).
\]
Then the alternating Beltrami equation with Beltrami coefficient $\mu$ is a critical point of the above energy.  
\end{theorem}
\begin{proof}
By taking a test function in the interior of constant orientation or near the boundary $\partial \Omega$, the critical point property in these regions is verified no different from the classical case. It now suffices to check the critical point property for the region near the singular set configuration. Since the number of cusp points is finite, it will not contribute to the integration on the singular set. Hence it suffices to work locally in a small neighborhood $U$ that contains only folding points, like the situation in Figure \ref{fig:folding_point}. Take any smooth test function $\phi$ compactly supported in $U$. Then by setting
\[
\frac{d}{d\epsilon}\big|_{\epsilon = 0}E_{GQC}(u+\epsilon \phi,v) = 0  
\]
we obtain
\[
\left (\int_{\Omega^+} - \int_{\Omega^-} \right) \langle P\nabla \phi, P\nabla u\rangle + \langle P\nabla \phi, JP\nabla v \rangle \, dxdy= 0.
\]
Integrating by parts, and repeating the same steps for $v$, we can derive the following Euler-Lagrange system
\begin{equation}
\begin{cases}
-\nabla\cdot(\tilde{A}\nabla u)  =0 \text{ in } \Omega^+\cup\Omega^- \\
-\nabla\cdot(\tilde{A}\nabla v)  =0 \text{ in } \Omega^+\cup\Omega^- \\
\partial_{\tilde{A}\nu}u +\nabla v \times \nu =0 \text{ on } \partial\Omega \cup \Sigma\\
\partial_{\tilde{A}\nu}v -\nabla u \times \nu =0 \text{ on } \partial\Omega \cup \Sigma
\end{cases},
\end{equation}
where 
\[
A= \frac{1}{1-|\mu|^2}\begin{bmatrix}(\rho-1)^{2}+\tau^{2} & -2\tau\\
-2\tau & (1+\rho)^{2}+\tau^{2}
\end{bmatrix},
\]
$\tilde{A} = A$ on $\Omega^{+}$ and $\tilde{A} = -A$ on $\Omega^{-}$, $\nu = \nu(z)$ is the outer unit normal vector when $z\in \partial\Omega^{+}$, or equivalently the inner unit normal if we regard $z\in \partial\Omega^{-}$. Note that the second order equation outside of the singular set and boundary has the same form, and corresponds to the alternating Beltrami equation with coefficient ${\mu}$. Most importantly, the boundary condition matches up from both sides of the folding line. This finishes the proof.
\end{proof} 

\begin{remark}
Note that the above Neumann boundary problem is somewhat different from convention since the singular set lies in the interior but is treated like boundary. But this very condition can  be seen as the way to glue two pieces of solutions on $\Omega^{+}$ and $\Omega^{-}$ together along the singular set configuration.
\end{remark}

\subsubsection{Discretization and implementation details}
First we discuss the case of least squares quasiconformal energy and later extend it to the generalized case. We discretize the the equation \eqref{eq: qc} on a linear triangular mesh $\mathcal{T}$, which is encoded as a list of vertices $V$ and a list of triangles $\mathcal{T}$ (by a mild abuse of notation) taking indices into $V$. We denote the number of vertices by $|V|$ and number of triangles by $|\mathcal{T}|$. The second order operator $\nabla\cdot(A\nabla)$ is a variant of the Laplacian. Its discretization amounts to expressing the following sum
\begin{equation} \label{eq: ele}
 \sum_{T\in\mathcal{T}}\langle P\nabla\varphi(T),P\nabla\phi(T)\rangle_{T}
\end{equation}
for any test functions $\varphi, \phi$ defined on the vertices $V$ into a quadratic form  $\varphi^T \mathcal{L}_{\mu} \phi$. Here, $\langle \cdot, \cdot \rangle_{T}$ is the 2D Euclidean inner product scaled by the area of the triangle $T$. On an oriented triangle $T = [v_0, v_1, v_2] $, since the functions being considered are linear on triangles, the gradient of a function 
$\varphi = (\varphi_0, \varphi_1, \varphi_2) $ on this triangle is given by
\[
\nabla\varphi=\frac{1}{2\text{Area}(T)}\begin{bmatrix}0 & -1\\ 1 & 0 \end{bmatrix}\sum_{i=0,1,2}\varphi_{i}(v_{2+i}-v_{1+i}).
\]
where indexing modulo $3$ as appropriate.
Observe that
\[
\begin{bmatrix}1-\rho & -\tau \\ -\tau & 1+\rho\end{bmatrix} \begin{bmatrix}0 & -1\\ 1 & 0 \end{bmatrix} = \begin{bmatrix}0 & -1\\ 1 & 0 \end{bmatrix} \begin{bmatrix}1+\rho & \tau \\ \tau & 1-\rho\end{bmatrix} 
\]
Hence, denoting
\[
v'_i = P^{-1}v_i,
\]
we have
\[
P\nabla\varphi=\frac{1}{2\text{Area}(T')}\begin{bmatrix}0 & -1\\ 1 & 0 \end{bmatrix}\sum_{i=0,1,2}\varphi_{i}(v'_{2+i}-v'_{1+i}).
\]
Therefore, denoting the triangle $T' = [v'_0, v'_1, v'_2]$,
\[
\begin{array}{rcl} \langle P\nabla\varphi(T),P\nabla\phi(T)\rangle_{T} & = & -\frac{1}{4\text{Area}(T')}\sum_{i,j}\varphi_{i}\phi_{j}(v'_{2+i}-v'_{1+i})^{T}(v'_{2+j}-v'_{1+j})\\ & = & -\sum_{i,j}\omega_{ij}(T)\varphi_{i}\phi_{j} \end{array}
\]
where 
\[
\omega_{ij}(T)=\begin{cases} -\frac{1}{2}\cot\theta'_{k},\,k\neq i,j & \text{if }i\neq j\\ \frac{1}{2}(\cot\theta'_{i+1}+\cot\theta'_{i+2}) & \text{if }i=j \end{cases} 
\]
where $\theta'_k$ is the angle of at the vertex $v'_k$. This is noting but the {\it cotangent weight} with angles changed by the effect of $\mu$. 

The expression for the area integral $\mathcal{A}(u,v)$ is unchanged from the least square conformal case \cite{levy2002least}, and hence we have:
\begin{corollary} \label{cor:1}
The quadratic form (up to a nonzero constant scaling) associated to the triangular mesh discretization of the least squares quasiconformal energy is given by the following $2|V| \times 2|V|$ matrix
$$
M := \text{diag}(\mathcal{L}_{\mu},\mathcal{L}_{\mu}) - 2\mathcal{A},
$$
which is applied to the $2|V|$-coordinate vector ${\bf x} = (u,v)$. The discrete version of Equation \eqref{eq: qc} is then $M\bf{x}=0$. Here $\mathcal{L}_{\mu}$ is the cotangent matrix associated to the operator $\nabla \cdot A\nabla$,
and $\mathcal{A}$ is the area matrix of the target triangular mesh.
\end{corollary}
\begin{proof}
It follows from the discussion above that the $|V| \times |V|$ matrix $\mathcal{L}_{\mu}$ corresponds to the discretization of the differential operator $\nabla\cdot(A\nabla)$. The area matrix matrix has non-zero entries only corresponding to the boundary vertices. It is then immediate to check, by examining the corresponding rows of the linear system, that for interior vertices, the solution $(u,v)$ satisfies $-\nabla\cdot(A\nabla u)  =0 $ and
$-\nabla\cdot(A\nabla v)  =0 $, while on the boundary, it satisfies $\partial_{A\nu}u +\nabla v \times \nu =0 $ and
$\partial_{A\nu}v -\nabla u \times \nu =0 $ respectively.
\end{proof} 

In order to obtain a nontrivial solution to the system $M\mathbf{x}= 0$, it turns out one need only pin down at least two vertices. The precise statement is contained in the following proposition.
\begin{proposition}
Suppose $|\mu|$ is uniformly bounded away from $1$, and the triangular mesh is connected and has no dangling triangles (i.e. there are no triangles which share a common vertex but no common edge). Let $I_{\text{pin}}$ be the indices of the points to be pinned, with cadinality $|I_{\text{pin}}|\geq 2$, $I_{\text{free}}$ be the indices for the free points, and $M$ be the matrix defined as in Corollary \ref{cor:1}. Then the $2|I_{\text{free}}| \times2|I_{\text{free}}|$ sub-matrix $M_{free}$ of $M$ indexed by the free points has full rank.
\end{proposition}
\begin{proof}
The idea is essentially identical to the proof in \cite{levy2002least} and we only sketch the main argument and the modification needed here. The key observation is that the (topological) triangular mesh satisfying our assumption can be constructed incrementally using two operations:
\begin{itemize}
    \item Glue: adding a new vertex and connecting it to two neighboring vertices;
    \item Join: joining two existing vertices.
\end{itemize} 
The proof is based on this observation. To wit, we express $E_\text{LSQC}$ as $\|B\mathbf{x} - \mathbf{b}\|^2$ where $B$ is of size $2|\mathcal{T}|\times 2|I_{free}|$. It then suffices to prove $B$ has full rank. One then proceeds by proving the incremental construction preserves the full rank property. Since the modulus of the Beltrami coefficient associated to the new triangle is bounded away from $1$, the associated matrix coefficients are nonzero real numbers. And therefore the same argument in \cite{levy2002least} applies.
\end{proof}
\begin{remark}
We now have a geometric understanding in the discrete case. In fact, it can be regarded as a conformal map with the domain given a different {\it conformal structure}. This viewpoint in fact has been already demonstrated previously when we derive the Beltrami equation. In the case of a triangular mesh, this structure can be thought of the assignment the angle triples $(\alpha_1,\alpha_2,\alpha_3)$ to each triangular face of the mesh, or equivalently the associated Beltrami coefficient. Under this viewpoint we can relate many algorithms from conformal geometric processing to their quasiconformal counterparts.  
\end{remark}
\begin{remark}
It is usually a preferred practice to choose these two points far apart from each other to reduce excessive local scale change, as is the case in the conformal flattening task \cite{levy2002least}. This is because the triangle angles associated to the Beltrami coefficients as given may not be realizable as a planar mesh, and the associated least square quasiconformal energy can never reach zero.
\end{remark}

The implementation for the generalized quasiconformal energy is essentially identical to the previous case, except for the separate treatments for the oriented preserving and reversing parts of the domain. 
This means when assembling the matrix using triangle elements as in \eqref{eq: ele}, for the orientation preserving elements it is identical to the previous case, while for the orientation reversing elements we need to multiply a minus sign.

\begin{remark} \label{rem:unfold}
An unfolding map satisfies also an alternating Beltrami equation, thus it can be computed in the same way, since in order to assemble the matrix $M$ we need only to know the coefficients on each triangular mesh element. In fact, the Beltrami coefficient for the inverse map at corresponding location can be verified to be the same using a composition formula \cite{astala2008elliptic} which is local in nature. The computation of unfolding map will be important in the reinforcement iteration introduced later.
\end{remark}

In the sequel, solving (alternating) Beltrami equations using the described method will be denoted as
\[
\mathcal{T'} = \text{LSQC}(\mathcal{T},(\mu_{T})_{T\in\mathcal{T}},\mathcal{C}),\]
where $\mathcal{T'}$, $\mathcal{T}$ are the computed target triangular mesh and the domain mesh, respectively; $(\mu_{T})_{T\in\mathcal{T}}$ is the set of Beltrami coefficients on each face, and $\mathcal{C}$ is the set of constraints
\footnote{The source code is available at \url{https://github.com/sylqiu/Least-square-beltrami-solver}, and it supports the alternating Beltrami equation. }.

\section{Recovering flat-foldable surface from self-occlusion and the reinforcement iteration}  \label{sec3}
In this section we formalize the recovery problem and propose its solution using the framework developed in the previous section. 
\subsection{Problem description}
\label{sec3.1}
Following the notations as before, let $\Omega_{\Sigma}^{+},$ $\Omega_{\Sigma}^{-}$
be the two disjoint open sets in $\Omega$ specified by some singular
set configuration $\Sigma$ as the orientation-preserving part and
the reversing part, respectively. One possible formulation is to find
a folding homeomorphism $f$ such that 
\[
\arg\min_{f,\Sigma} E_{\text{GQC}}(u, v, \mu) 
\]
where $f = (u,v)$, 
\begin{equation} \label{eq:14} 
\mu =  \begin{cases} 
 0 &\text{ in }\Omega_{\Sigma}^{+}\\
 \infty &\text{ in }\Omega_{\Sigma}^{-}
\end{cases}
\end{equation}
subject to the constraints
\[
f\big|_{\Omega_{\text{vis}}}:\Omega_{\text{vis}}\to S_{\text{vis}},
\]
where $S_{\text{vis}}$ denote a subsets of the boundary and singular set that is not occluded and $\Omega_{\text{vis}}$ the corresponding points in the encoding domain.

Note that $\Sigma$ is a variable in the minimization problem. This
formulation as it stands seems very hard to implement. To resolve the difficulty, we ask for a factorization $f \circ \varphi = g$, where $f$ is flat-foldable, $g$ is some reference folding map that satisfies the constraints, and $\varphi$ is quasiconformal and preserves the constraints. 
It is in a spirit similar to the classical Stoilow factorization \cite{astala2008elliptic}.
Specifically, we introduce an initial singular set configuration $\Sigma_{0}$ of the same topology:
\begin{equation}
\arg\min_{f,\varphi}\left(E(f,\varphi):=E_{\text{GQC}}(u, v, \mu) \right)\label{eq:6}
\end{equation}
where $ g =(\tilde{u}, \tilde{v})$ are such that $\tilde{u}= u \circ \varphi, \tilde{v} = v \circ \varphi$, $\mu = 0$ on $\Omega_{\varphi(\Sigma_0)}^{+}$ and $\mu = \infty$ on $\Omega_{\varphi(\Sigma_0)}^{-}$ the minimization is subject to the constraints
\[
f \big|_{\Omega_{\text{vis}}}:\Omega_{\text{vis}}\to S_{\text{vis}}.
\]
Now the argument $f$, defined on $\Omega$, ranges in the set of folding
homeomorphisms, and $\varphi:\Omega\to\Omega$ ranges in the set
of quasiconformal homeomorphisms.
Note that if $\varphi$ gives the ``correct'' singular set configuration,
then the above energy vanishes for $f$ that is generalized conformal. 
Note that there might be multiple solutions satisfying the same partial constraints. To formalize the situation we define the admissibility of solutions described below in Definition \ref{def:admissible}.

\subsection{The reinforcement iteration}
The iteration consists of two key steps, which find improved folding and unfolding maps given the current unfolded and folded surfaces in an alternating fashion. Intuitively, each unfolding step tries to find a better singular set configuration, while each folding step tries to conform with the given data.

To motivate, let us first consider the case where the entire boundary and singular set of the folded surface $S$ is given. In order to parametrize $S$, 
one needs to start with some initial singular set configuration. We can easily construct a map $g:\Omega_{\Sigma_0} \to S$ by enforcing all the correct singular set and boundary constraints.

In general, $g$ will not be a conformal map, as the initial singular set configuration may not coincide with the reality. So instead, there exists a quasiconformal
map 
\[
\varphi:\Omega_{\Sigma_0}\to\Omega_{\Sigma}
\]
that maps the initial configuration to the correct one. Its relation with the desired generalized conformal map $f$ can be observed
as a commutative diagram below
\begin{equation}
\begin{tikzcd} \label{diagram:1}
\Omega_{\Sigma_0} \arrow[r,"\varphi"] \arrow[d,"g"] & \Omega_{\Sigma} \arrow[ld,"f"]\\
S 
\end{tikzcd}
\end{equation}
In this case, once we obtain the folded surface $S$ from the map $g$, since the entire boundary and singular set data is given, the generalized conformal unfolding homeomorphism $f^{-1}$ can be constructed by solving the alternating Beltrami equation with 
\begin{equation*}  
\nu =  \begin{cases} 
 0 &\text{ in }S^{+}\\
 \infty &\text{ in }S^{-}
\end{cases}
\end{equation*}
And in this way the map $\varphi$ is obtained by the composition $f^{-1}\circ g$. \\

However, when only partial data of $S$ is provided, it is no longer possible to obtain the folded surface $S$ by constructing $g$ in the above manner. We need to find the folded surface and its Euclidean singular set configuration simultaneously. We shall be looking for a folded surface that satisfies the following properties.
\begin{definition} \label{def:admissible}
Let $S_{\text{vis}}$ be a set of partial boundary and singular set data, and $\Omega_{\Sigma_{0}}$ be the domain with an initial singular set configuration.  A folded surface $S$ is called admissible if 
\begin{enumerate}
\item Topological equivalence: The singular set configuration of $S$ is of the same topological type with the target surface.
\item Data correspondence: There is a subset $C\subset S$ such that there is a isometry from $C$ to $S_{\text{vis}}$.
\item Cycle consistency: There exist maps $g,\varphi,f$ such that $f$ is flat-foldable, and the diagram \eqref{diagram:1} commutes.
\end{enumerate}
\end{definition}

We now proceed to describe a fixed-point-like algorithm of finding some admissible folded surface $S$ and its parametrization simultaneously.  Let $n$ be the current iteration number, $\Omega_{\Sigma_n} = \Omega_{\varphi_n(\Sigma_0)}$, and
$
g_{n}:\Omega_{\Sigma_{n-1}}\to S_{n}
$
be a generalized quasiconformal folding homeomorphism that satisfies
the constraints 
\[
g_{n}\big|_{\Omega_{\text{vis}}}:\Omega_{\text{vis}}\to S_{\text{vis}},
\]
where $\Omega_{\text{vis}}\subset\Omega_{\Sigma_(n-1)}$ is the corresponding subset corresponding to the partial boundary and singular set data $S_\text{vis}$. This step promotes data fidelity. Let $h_{n}:S_{n}\to\Omega_{\Sigma_{n}}$ be unfolding homeomorphism, which is obtained by solving the alternating Beltrami equation with 
\begin{equation} \label{eq:nu}
\nu = \begin{cases} 
 0 &\text{ in }S_{n}^{+}\\
\infty &\text{ in }S_{n}^{-}
\end{cases}
\end{equation}
with enforcing the original shape constraints of $\Omega_{\Sigma_{n}}$. This step minimize the generalized conformal distortion based on fitted surface $S_n$. We note that both the $S_{\text{vis}}$ and shape constraints are essential for convergence of the algorithm, in particular implicitly decreasing the area distortion of the map. The next map $g_{n+1}$ is constructed based on the
updated domain with its singular set configuration 
\[
\Sigma_{n}=h_{n}\circ g_{n}(\Sigma_{n-1}).
\]
As $n\to \infty$, we want $g_{n}$ to converge to a generalized conformal map $f:\Omega\to S$, and the composition
\[
h_{n}\circ g_{n}:=\psi_{n}
\]
converges to $id_{\Omega}$, while 
$$
\psi_{n}\circ\cdots\circ\psi_{2}\circ\psi_{1}:=\varphi_n
$$ converges to a quasiconformal map that transforms
the initial singular set configuration to a desirable one. This is shown
schematically in the following diagram
\begin{equation}
\begin{tikzcd} \label{diagram:2}
\Omega_{\Sigma_{0}} \arrow[r,"\psi_{1}"] \arrow[d,"g_1"] 
& \Omega_{\Sigma_{1}} \arrow[r,"\psi_{2}"] \arrow[d,"g_2"]  
& \Omega_{\Sigma_{2}} \arrow[r,"\psi_{3}"] \arrow[d,"g_3"] 
& \cdots \Omega_{\Sigma_{n-1}} \arrow[r,"\psi_{n}"] \arrow[d,"g_n"]  
& \Omega_{\Sigma_{n}} \cdots\\
S_{1} \arrow[ru,"h_1"] 
& S_{2} \arrow[ru,"h_2"] 
& S_{3} \arrow[ru,"h_3"]
& S_{n} \arrow[ru,"h_n"]
\end{tikzcd}
\end{equation}
In each step we keep enforcing the available data $S_{\text{vis}}$ by the map $g_n$, and by $h_n$ we keep enforcing the known boundary shape of $\Omega$, hence the name reinforcement iteration. The overall algorithm is summarized as in Algorithm \ref{algo:1}.

\begin{algorithm}
\caption{Reinforcement Iteration}

{\bf Inputs: }Initialized domain $\Omega_{\Sigma_0}$, partial data from folded surface $S_{\text{vis}}$, termination threshold $\epsilon>0$.\\
{\bf Outputs :} Optimal mapping $\varphi^*$ such that $\varphi^*(\Omega_{\Sigma_0})$ has optimal singular set configuration.

\vspace{2pt}
\hrulefill
\vspace{2pt}
\begin{algorithmic}
\State{Construct $g_{1}$, $h_1$; compute $\varphi_{1} = u_{1}\circ f_{1} $.}
\State{Evaluate $E(g_{1},\varphi_{1})$, $E(g_{0},\varphi_{0})=0$, $n=1$.}
\While{$|E(g_{n},\varphi_{n})-E(g_{n-1},\varphi_{n-1})| > \epsilon$,} 
\State{Provided $\varphi_{n-1}$, construct $g_{n}$, $h_n$ using the partial data and domain shape constraints.} 
\State{Compute $\varphi_{n} = h_{n}\circ g_{n}\circ \varphi_{n-1}$.} 
\State{Evaluate $E(g_{n},\varphi_{n})$, $n\leftarrow n+1$.}
\EndWhile 
\end{algorithmic}  \label{algo:1}
\end{algorithm}

The basic steps are constructions of the maps
$g_{n}$ and $h_{n}$, whose implementations we now turn to.
\subsection{Implementation details}

\subsubsection{Construction of \texorpdfstring{$g_{n}$}{}}

Given the updated domain $\Omega_{\Sigma_{\varphi_{n-1}}}=\varphi_{n-1}(\Omega_{\Sigma_{0}})$,
we obtain the $S_{\text{vis}}$-enforced map $g$ by solving a alternating Beltrami equations subject to the constraints
\[
g\big|_{\Omega_{\text{vis}}}:\Omega_{\text{vis}}\to S_{\text{vis}}.
\]
Since we mainly aim for the partial boundary and singular data enforcement here, in all examples shown we simply set the Beltrami coefficients to be $0$ or $\infty$ on the orientation-preserving or -reversing regions, respectively.
This choice of coefficients is a good one if the initialized domain is close to an admissible one and it will become increasingly so provided the algorithm converges and $g_n$ close to generalized conformal. Note that if for some sub-domain the mapping is already known, then the corresponding part of Beltrami coefficients can of course be pre-computed in the first iteration. Then we can continue to set the Beltrami coefficients as before since the previous part of distortion has been ``factored out".

In terms of triangular meshes, suppose the domain triangular mesh is
$\mathcal{D}_{n-1},$ with $\mathcal{D}_{n-1}^{+}$,
$\mathcal{D}_{n-1}^{-}$ corresponds to $\Omega_{n-1}^{+}$
and $\Omega_{n-1}^{-}$, respectively; $S_{\text{vis}}$
is realized as certain constraint $C_\text{vis}$. Then the folded surfaced is obtained by 
\[
\mathcal{S}_{g_{n}}=\text{LSQC}(\mathcal{D}_{n-1},\{\mu_{T}\}_{\mathcal{D}_{\varphi_{n-1}}},C_\text{vis})
\]
where \footnote{In practice we take a complex number with large enough modulus (e.g. $\mu = 10^5$) in place of the $\infty$.}
\[
\mu_{T}=\begin{cases}
0 & \text{if }T\in\mathcal{D}_{n-1}^{+}\\
\infty & \text{if }T\in\mathcal{D}_{n-1}^{-}
\end{cases}.
\]

\subsubsection{Construction of \texorpdfstring{$h_{n}$}{}}

Given the folded surface constructed from the last step $S_{n}=g_{n}(\Omega_{n-1})$,
recall that the unfolding map $h_{n}:S_{n}\to\Omega_{\Sigma_{n}}$ is found by solving the minimization problem 
\begin{equation} \label{eq: h}
\arg\min_{h:S_n \to \Omega_{\Sigma_{n}}} E_{\text{GQC}}(h; \nu) 
\end{equation}
with $\nu$ as in \eqref{eq:nu}, and subject to the shape constraints
\[
h\big|_{\partial S_{n}}:\partial S_{n}\to \partial \Omega.
\]
$\varphi_{n}$ is then updated by $\varphi_{n} = h_{n}\circ g_{n}$.
In terms of triangular meshes, suppose the folded surface mesh is
$\mathcal{S}_{n},$ with $\mathcal{S}_{n}^{+}$,
$\mathcal{S}_{n}^{-}$ corresponds to $S_{n}^{+}$
and $S_{n}^{-}$, respectively; $\partial\Omega$
is realized as certain constraint $C_{\partial\Omega}$. Then the above minimization
can be solved by 
\[
\mathcal{D}_{\Sigma_{n}}=\text{LSQC}(\mathcal{S}_{n},\{\nu_{T}\}_{\mathcal{S}_{n}},C_{\partial\Omega})
\]
where 
\[
\nu_{T}=\begin{cases}
0 & \text{if }T\in\mathcal{S}_{n}^{+}\\
\infty & \text{if }T\in\mathcal{S}_{n}^{-}
\end{cases}.
\] 

\section{Further discussion and experimental results} \label{sec4}
We note that the desirable domain $\Omega^*$ is a fixed point of our iteration algorithm, which we write as $\Omega^* = \mathcal{F}(\Omega^*)$,
where $\mathcal{F}$ is the iteration map in operator form. Note that $\mathcal{F}$ depends on its argument $\Omega$ in a very non-linear way because of the intermediate folded surface $S_{n}$ we introduced in each iteration, whose computation requires the cotangent matrix associated to $\Omega_{\Sigma_n}$. But approximately, in each iteration the folding and unfolding operations are inverse to each other and therefore $\mathcal{F}$ is close to the identity. The convergence of fixed-point iteration is well studied in the literature, see \cite{combettes2004solving} and references therein. For example, the convergence will be implied by the $\alpha$-averaged property of $\mathcal{F}$.  As we can notice in Figure \ref{fig:iter_experiment}, as well as in many other experiments, the distortion of many of the interior mesh triangles can barely be noticed in the later phase of the iteration, while the meshes remain also well conditioned. As other fixed-point iterations, it is reasonable to expect that the iteration map under good conditioning of the mesh triangles, enough constraints and a good initial guess to have convergence. 

\begin{figure} 
    \centering
    \begin{subfigure}[b]{0.29\textwidth}
        \includegraphics[width=\textwidth]{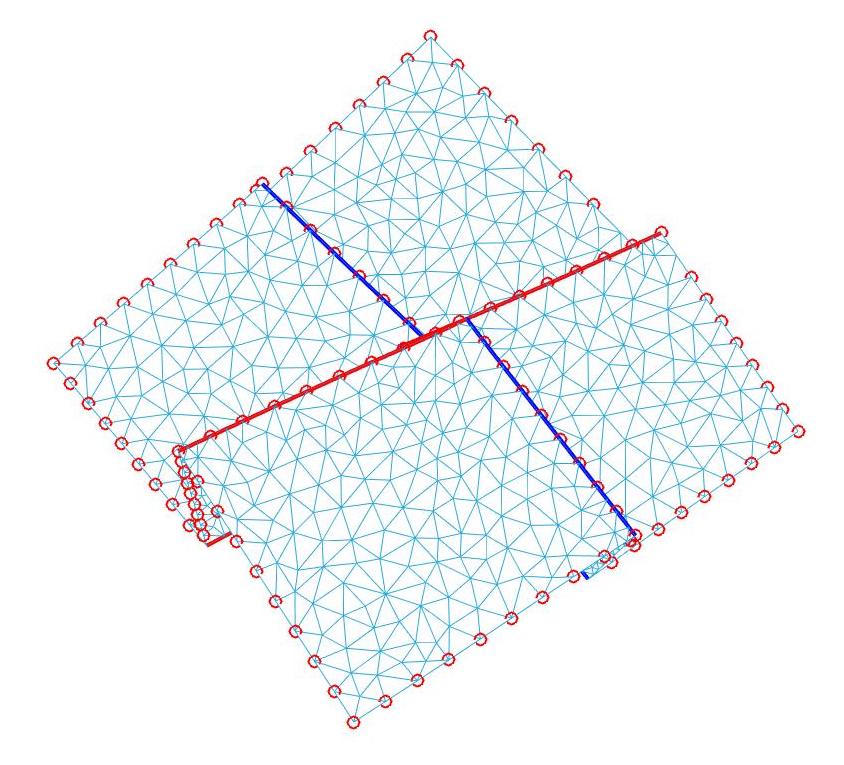}
        \caption{Folded surface with occlusion.}
        \label{fig:observed_fold}
    \end{subfigure} \quad
    \begin{subfigure}[b]{0.29\textwidth}
        \includegraphics[width=\textwidth]{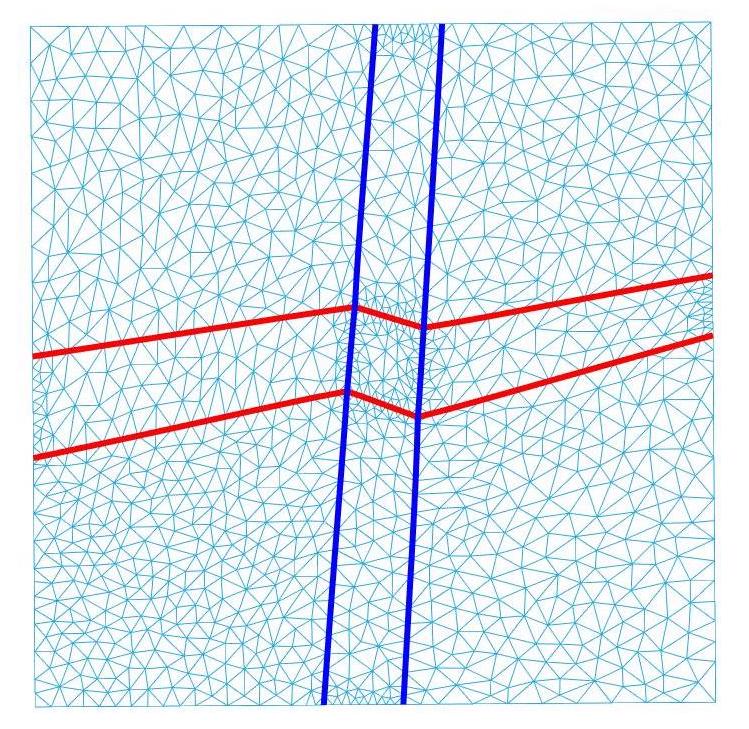}
        \caption{True unfolded surface.}
        \label{fig:true_unfold}
    \end{subfigure}
    \begin{subfigure}[b]{0.26\textwidth}
        \includegraphics[width=\textwidth]{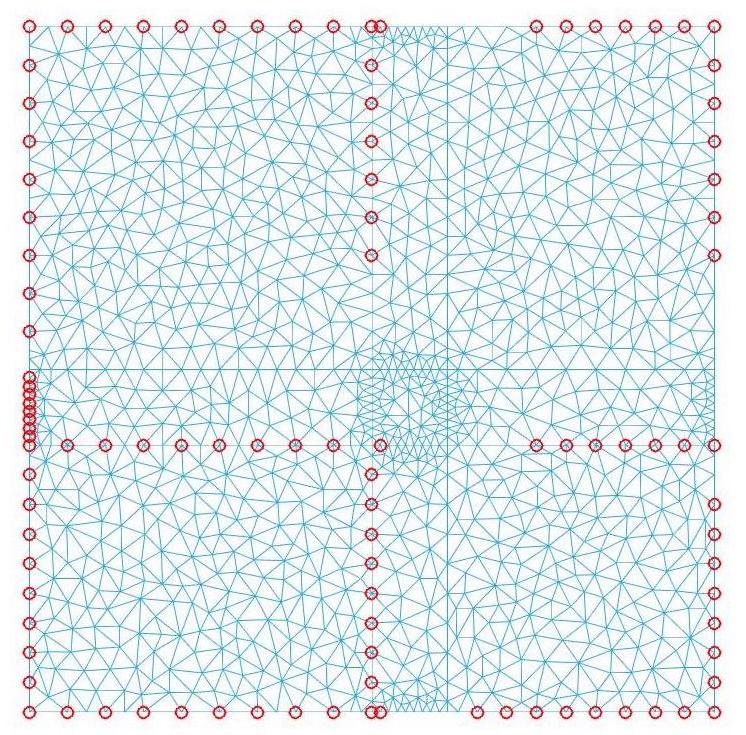}
        \caption{Initialised domain with partial data.}
        \label{fig:2unfold_init}
    \end{subfigure} \\
    \begin{subfigure}[b]{0.28\textwidth}
        \includegraphics[width=\textwidth]{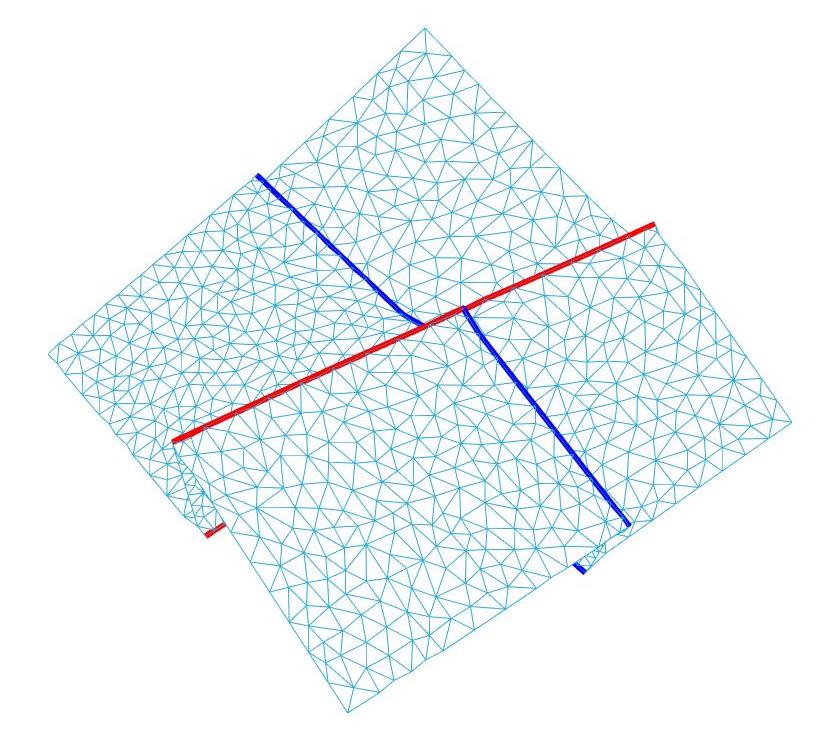}
        \caption{Iter = 1: Frontside of registered fold.}
        \label{fig:iter=1 reg+}
    \end{subfigure}
    \begin{subfigure}[b]{0.28\textwidth}
        \includegraphics[width=\textwidth]{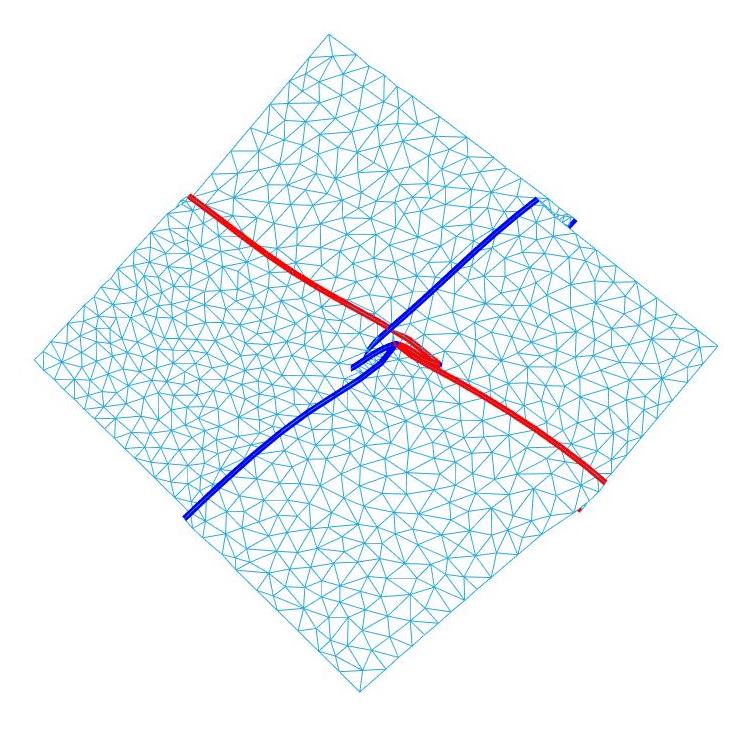}
        \caption{Iter = 1: Backside of registered fold.}
        \label{fig:iter=1 reg-}
    \end{subfigure}
    \begin{subfigure}[b]{0.28\textwidth}
        \includegraphics[width=\textwidth]{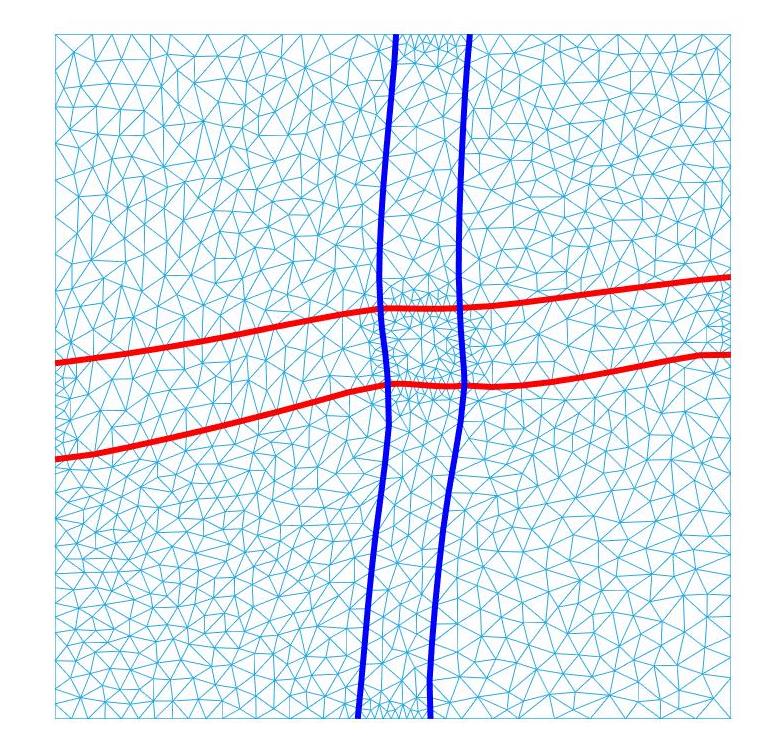}
        \caption{Iter = 1: Unfolded surface.}
        \label{fig:iter=1 unfold}
    \end{subfigure} \\
     \begin{subfigure}[b]{0.28\textwidth}
        \includegraphics[width=\textwidth]{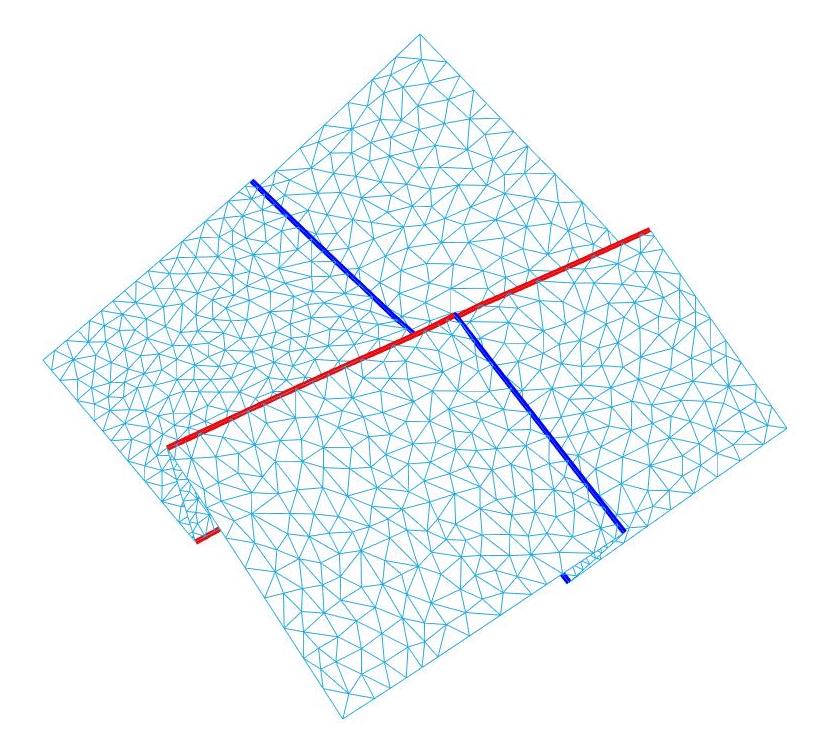}
        \caption{Iter = 50: Frontside of registered fold.}
        \label{fig:iter=5o reg+}
    \end{subfigure}
    \begin{subfigure}[b]{0.28\textwidth}
        \includegraphics[width=\textwidth]{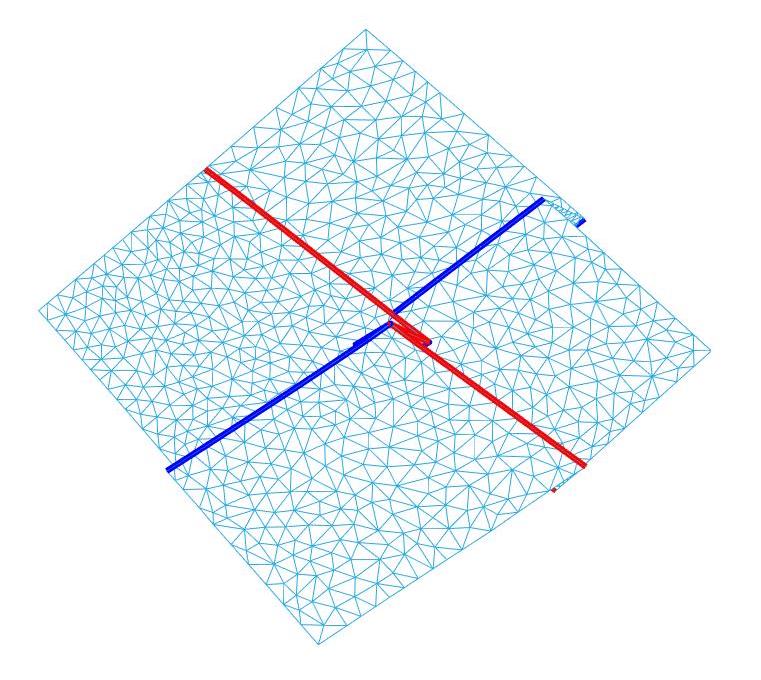}
        \caption{Iter = 50: Backside of registered fold.}
        \label{fig:iter=50 reg-}
    \end{subfigure}
    \begin{subfigure}[b]{0.28\textwidth}
        \includegraphics[width=\textwidth]{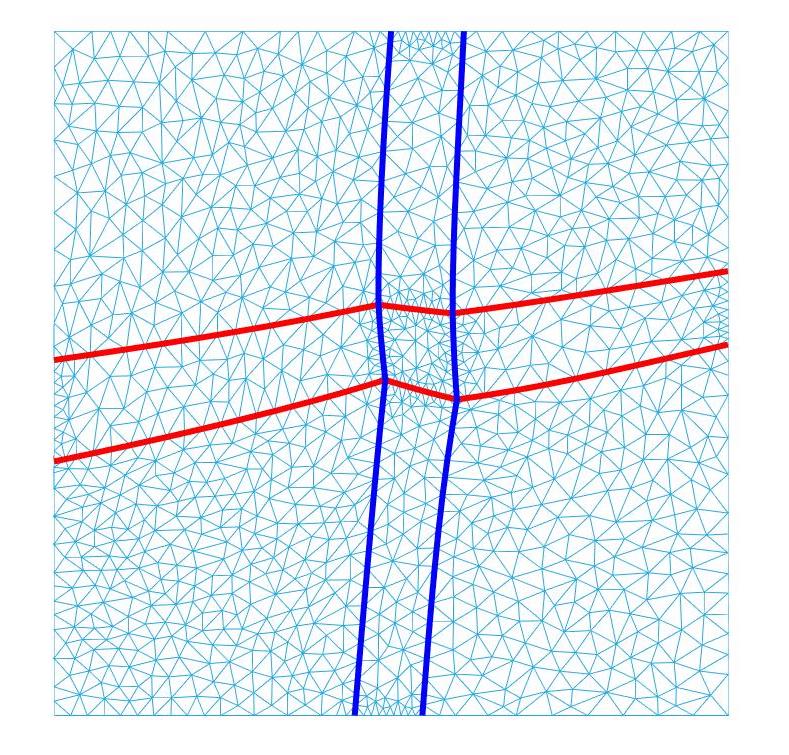}
        \caption{Iter = 50: Unfolded surface.}
        \label{fig:iter=50 unfold}
    \end{subfigure} \\
     \begin{subfigure}[b]{0.28\textwidth}
        \includegraphics[width=\textwidth]{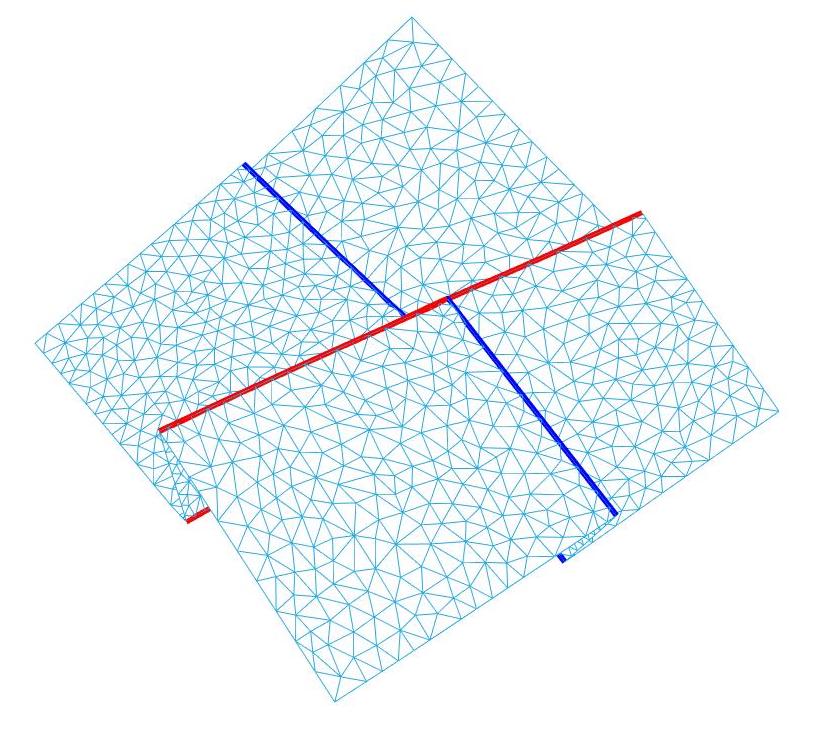}
        \caption{Iter = 200: Frontside of registered fold.}
        \label{fig:iter=120 reg+}
    \end{subfigure}
    \begin{subfigure}[b]{0.28\textwidth}
        \includegraphics[width=\textwidth]{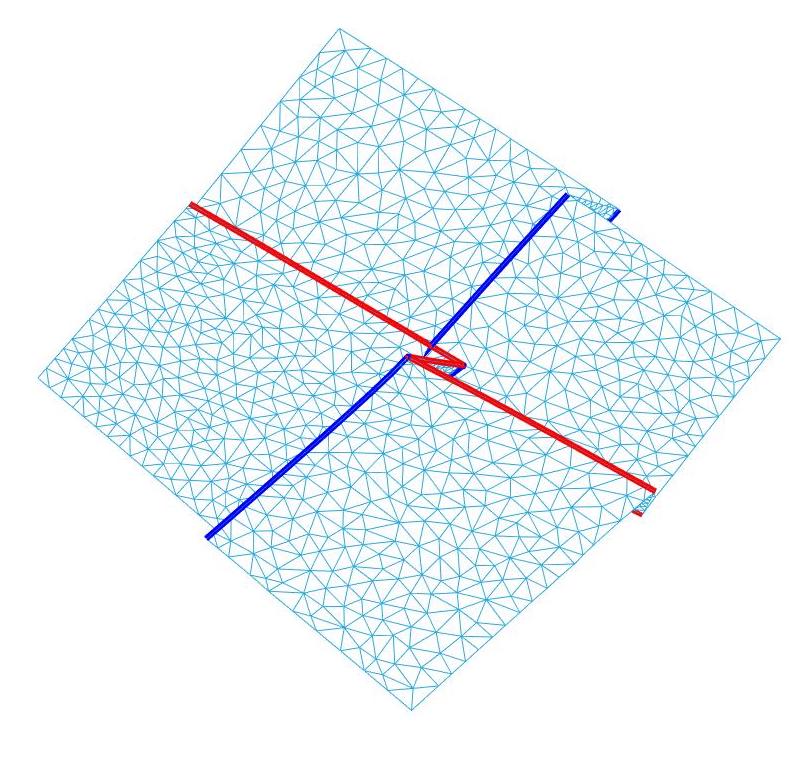}
        \caption{Iter = 200: Backside of registered fold.}
        \label{fig:iter=120 reg-}
    \end{subfigure}
    \begin{subfigure}[b]{0.28\textwidth}
        \includegraphics[width=\textwidth]{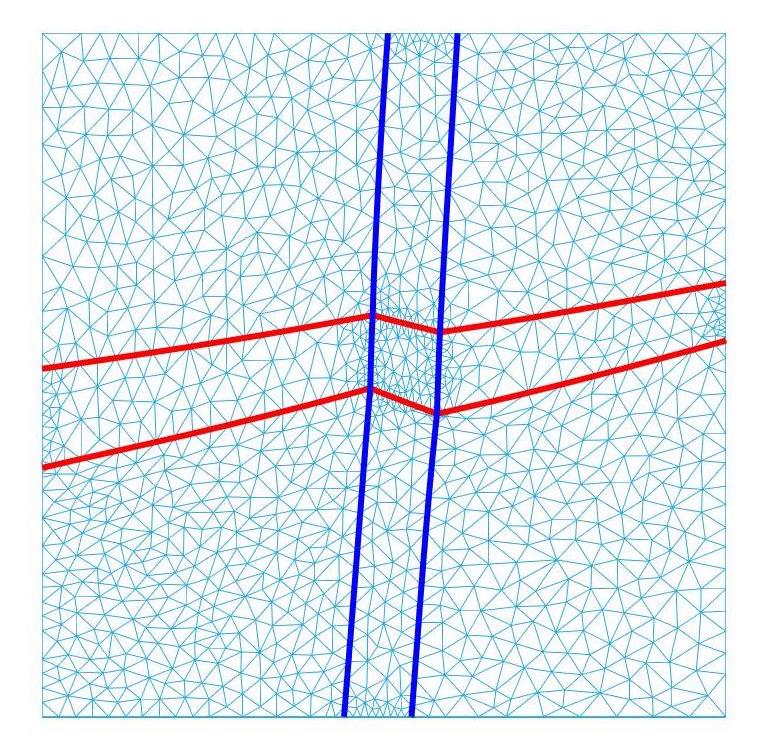}
        \caption{Iter = 200: Unfolded surface.}
        \label{fig:iter=120 unfold}
    \end{subfigure} \\
    \caption{Iteration results for the doubly folded surface: note that the folding lines gradually straighten out.}\label{fig:iter_experiment}
\end{figure}

\begin{figure}
    \centering
    \begin{subfigure}[b]{0.28\textwidth}
        \includegraphics[width=\textwidth]{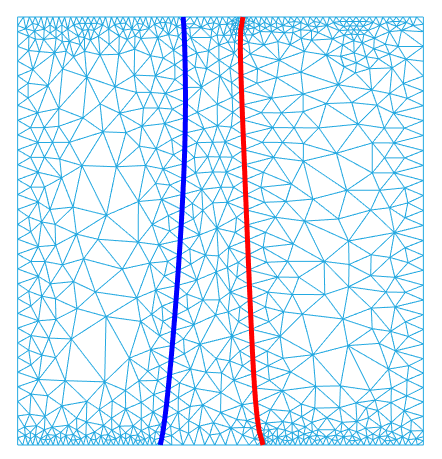}
        \caption{Iter = 1: unfolded domain.}
        \label{fig:1unfold_init}
    \end{subfigure}
    \begin{subfigure}[b]{0.285\textwidth}
        \includegraphics[width=\textwidth]{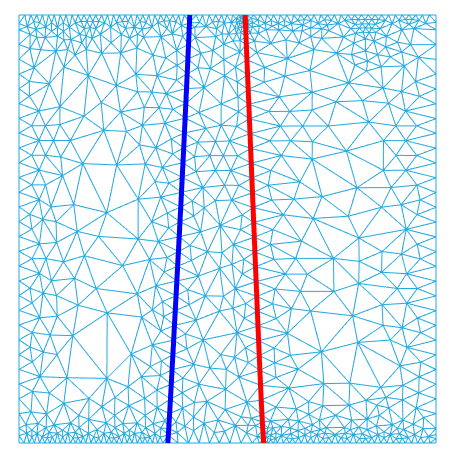}
        \caption{Iter = 50: unfolded domain.}
        \label{fig:1unfold_50}
    \end{subfigure}
    \begin{subfigure}[b]{0.295\textwidth}
        \includegraphics[width=\textwidth]{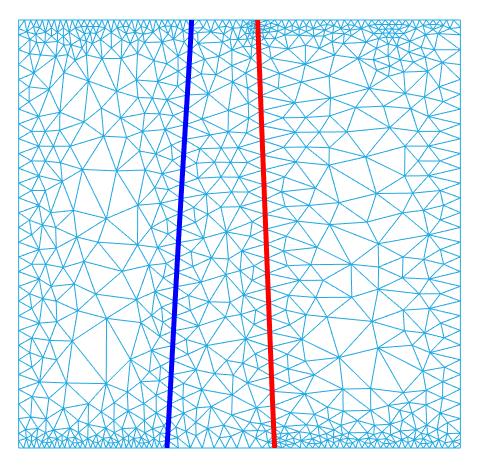}
        \caption{Iter = 200: unfolded domain.}
        \label{fig:1unfold_200}
    \end{subfigure}
    \begin{subfigure}[b]{0.28\textwidth}
        \includegraphics[width=\textwidth]{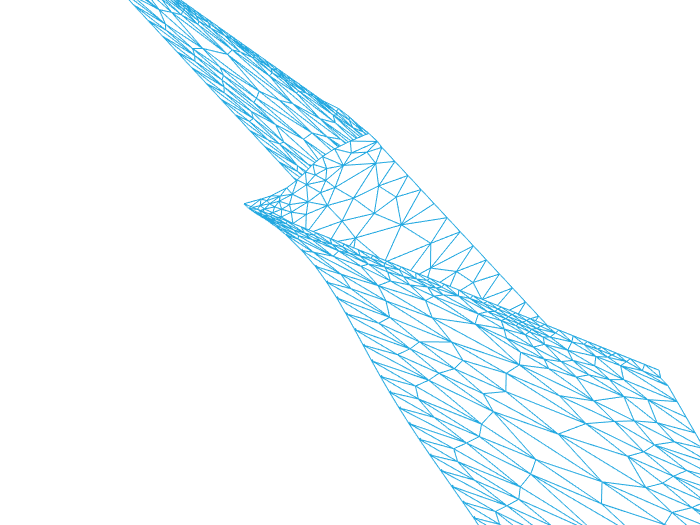}
        \caption{Iter = 1: initially registered 1-fold.}
        \label{fig:1fold_init}
    \end{subfigure}
    \begin{subfigure}[b]{0.28\textwidth}
        \includegraphics[width=\textwidth]{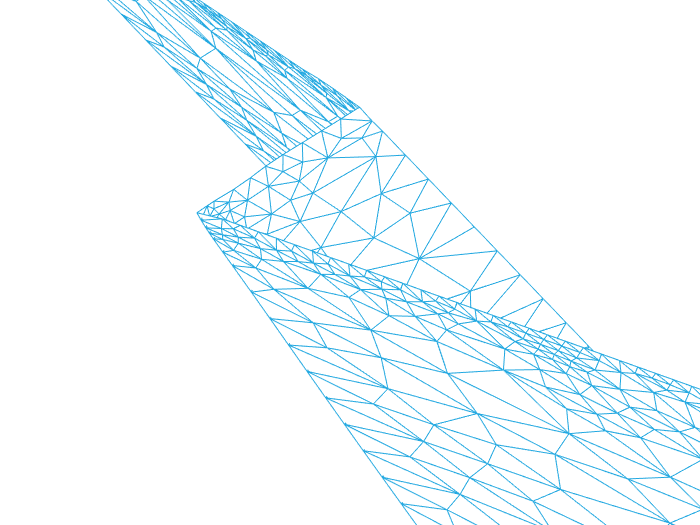}
        \caption{Iter = 200: finally registered 1-fold.}
        \label{fig:1fold_final}
    \end{subfigure}
    \caption{Iteration results for a once-folded surface: note the curved boundary in occlusion from the initial map is gradually straightened out and the folded domain becomes wider.}
    \label{fig:1fold_iter}
\end{figure}
\begin{figure}
\centering
    \begin{subfigure}[b]{0.28\textwidth}
        \includegraphics[width=\textwidth]{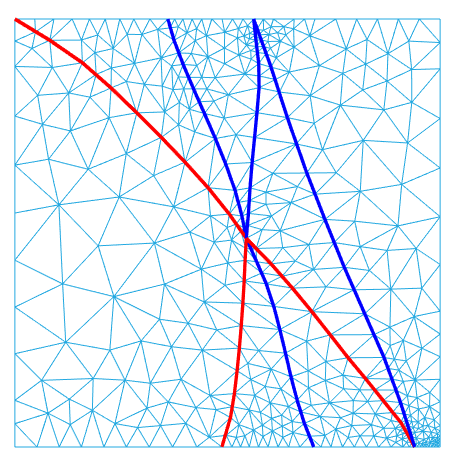}
        \caption{Iter = 1: unfolded domain.}
    \end{subfigure}
    \begin{subfigure}[b]{0.28\textwidth}
        \includegraphics[width=\textwidth]{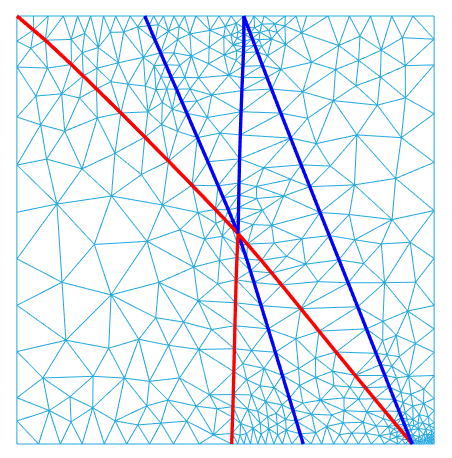}
        \caption{Iter = 50: unfolded domain.}
    \end{subfigure}
    \begin{subfigure}[b]{0.28\textwidth}
        \includegraphics[width=\textwidth]{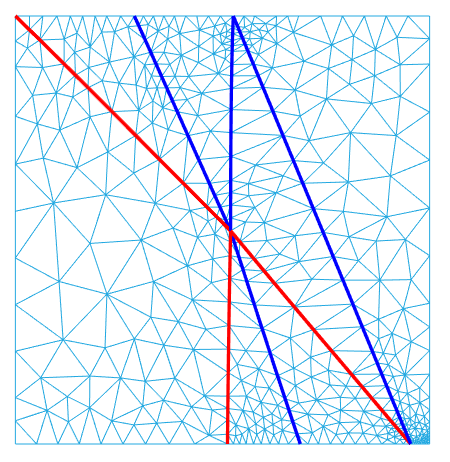}
        \caption{Iter = 200: unfolded domain.}
    \end{subfigure}
    \begin{subfigure}[b]{0.28\textwidth}
        \includegraphics[width=\textwidth]{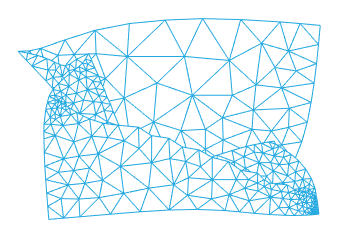}
        \caption{Iter = 1: initially registered cusp.}
    \end{subfigure}
    \begin{subfigure}[b]{0.29\textwidth}
        \includegraphics[width=\textwidth]{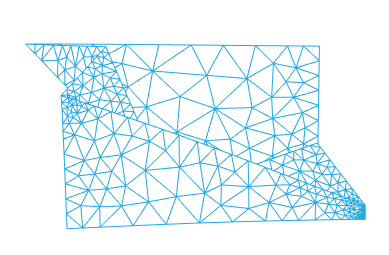}
        \caption{Iter = 200: finally registered cusp.}
    \end{subfigure}
    \caption{Iteration results for a cusped surface: note that the curved boundary and folding lines in occlusion from the initial map is gradually straightened out.}
    \label{fig:cusp_iter}
\end{figure}

We implemented the described reinforcement iteration algorithm and demonstrate it for a doubly folded surface, as illustrated in Figure \ref{fig:iter_experiment}. The folded surface and its unfolded counterpart, as shown in Figure \ref{fig:observed_fold} and \ref{fig:true_unfold}, are generated according to a real folded paper and its unfolded counterpart. In Figure \ref{fig:2unfold_init} it is our initialized domain $\Omega_{\Sigma_0}$. In \ref{fig:observed_fold} and \ref{fig:2unfold_init}, the red circles mark the corresponding constraint points to the visible partial singular set and boundary data. Our algorithm works similarly well with other examples as well. This shows the robustness of our algorithm.

In the next three rows of Figure \ref{fig:iter_experiment} we show the iteration results at the first iteration, 50-th iteration and 200-th iteration.  We can observe the curly folding lines in the first iteration in Figure \ref{fig:iter=1 reg-} and \ref{fig:iter=1 unfold}. This is due to the incomplete data and the incompatible initialized domain. In the subsequent iterations we saw significant improvement over the rigidity of the folding. In practice we also found that if in the later phase of the iteration, we explicitly regularize the singular lines by fitting the vertices into a Euclidean geodesic in a least square sense, and then restart the iteration, the convergence will have some minor speed-up in particular for the heavily multiply-folded cases.

Observe also that in the limit, as in Figure \ref{fig:iter=120 unfold}, the singular set configuration is in not exactly the same as that of the true unfolded surface. This can be explained by the existence multiple admissible solutions to this problem. For example, another admissible solution may be obtained by some different initialization. This is of course expected.

In Figure \ref{fig:1fold_iter} and \ref{fig:cusp_iter} we illustrate the effect of reinforcement iteration algorithm applied to a once-folded surface and a cusped surface. The straightening effect can be easily seen from the comparison between the initial folding map and the final folding map. In Figure \ref{fig:convergence}, we plot a log-log diagram for the scale insensitive version of the energy $E(g_k,\varphi_k)$ for the examples in Figure \ref{fig:iter_experiment} and \ref{fig:cusp_iter}. The plot here uses only the iteration algorithm and no other regularization. We can observe that the convergence rate approaches $O(1/N)$ in the mid-stage of the iteration. That the energy decreases slightly slower in the later phase can be explained by our observation from the iterations that only a few points are adjusted while the singular set configuration is still away from flat-foldability. These adjusted points are mainly near the cusp points. This fact can be observed from Figure \ref{fig:iter=1 unfold} and \ref{fig:iter=50 unfold}. The convergence rate varies in the different phases of the iteration, illustrating the non-linear nature of the iteration.

\begin{figure}
    \centering
    \includegraphics[width=0.7\textwidth]{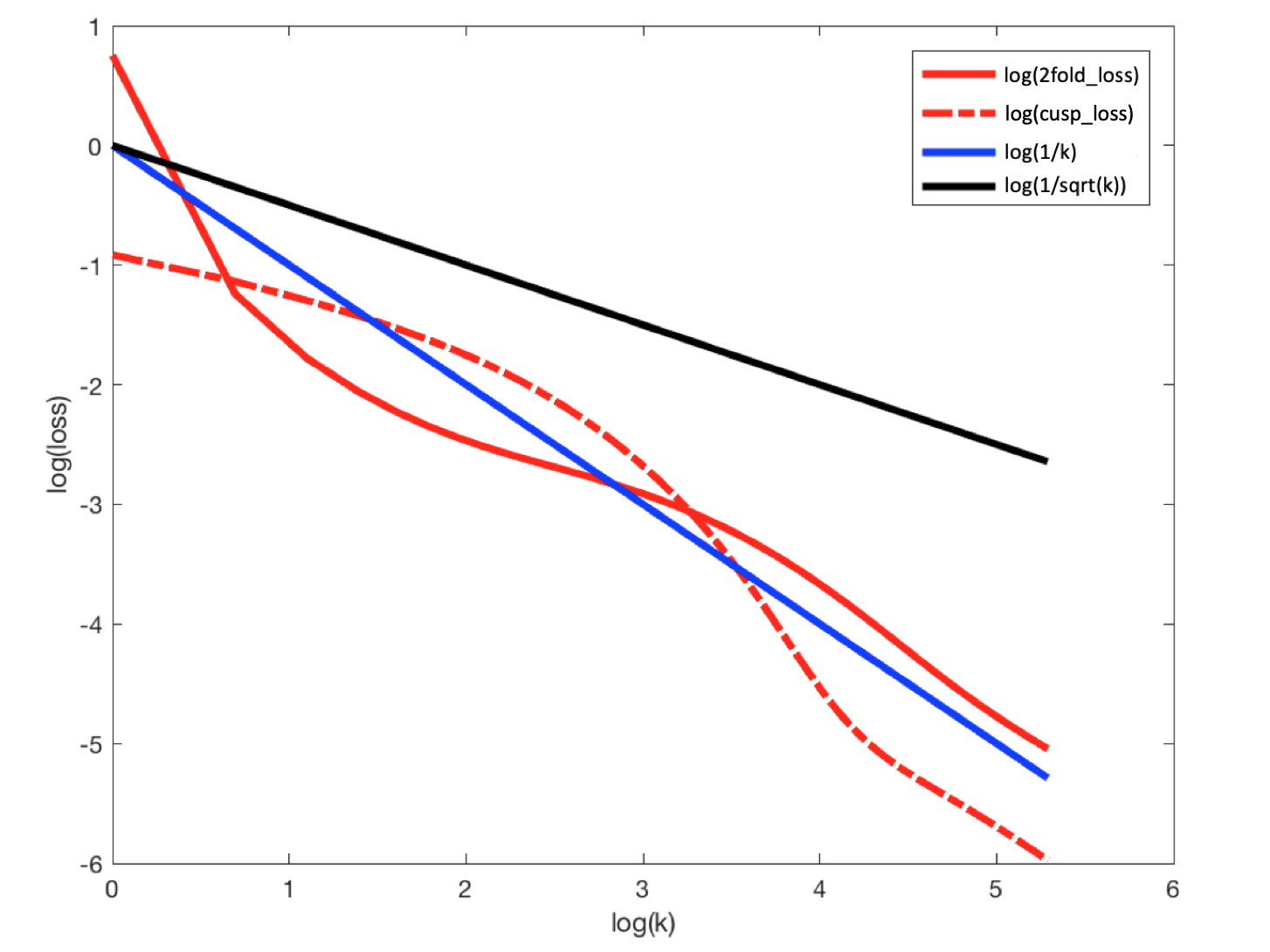}
    \caption{Convergence plot: the loss here is defined as $\sum_{T\in \Omega_{+}} |\mu_{g_n}|^2 +\sum_{T\in \Omega_{-}} 1/|\mu_{g_n}|^2 $ for scale invariant comparison, where $\mu_{g_n}$ denotes the Beltrami coefficient associated to the mapping $g_n$, which should be close to $0$ (or $\infty$) on $\Omega^{+}$ (or respectively  on $\Omega^{-}$).}
    \label{fig:convergence}
\end{figure}

\section{Applications} \label{sec5}

\subsection{Generating and editing generalized Miura-ori} \label{sec: miura}
The Miura-ori refers to a special type of Origami tessellation of the plane, which can be used to design flat-foldable materials aiming at achieving designed curvature properties \cite{dudte2016programming}. Previous approaches are based on analytic construction or constrained optimization, using the Kowasaki condition. Here we explore another possibility of creating such Origami models. Namely, we create more Miura-ori type domains and realize them via solving alternating Beltrami equations. 

For simplicity, we consider the Miura-ori pattern in Figure \ref{fig:miura}. The yellow color on a triangle $T$ refers to the prescription of $\mu(T) = \infty$, and purple ones $\mu(T) = 0$. After solving the alternating Beltrami equation in 2D by pining two vertices, we obtain the classical Miura-ori strip, which is the flat-folded state of the surface.  Suitable $z$-coordinates are added for visualization in 3D.

\begin{figure}
    \centering
    \begin{subfigure}[b]{0.33\textwidth}
        \includegraphics[width=\textwidth]{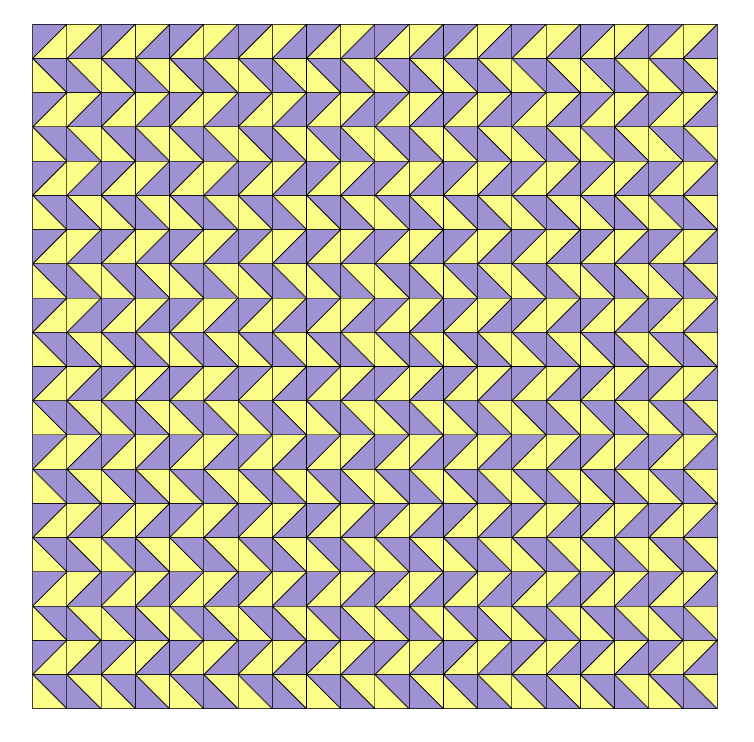}
        \caption{Classical Miura-ori pattern}
        \label{fig:miura_ori_unfolded}
    \end{subfigure}
    \hspace{10pt}
    \begin{subfigure}[b]{0.45\textwidth}
        \includegraphics[width=\textwidth]{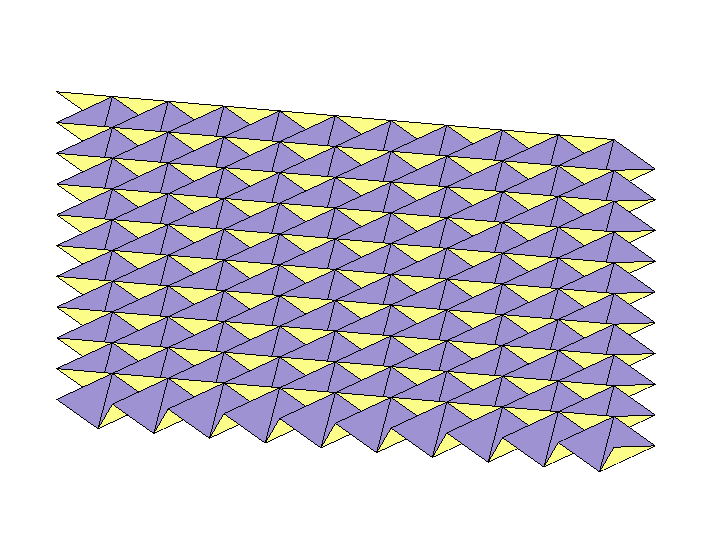}
        \caption{Realization of the Miura-ori on the left}
        \label{fig:miura_ori_folded}
    \end{subfigure}
    \begin{subfigure}[b]{0.39\textwidth}
        \includegraphics[width=\textwidth]{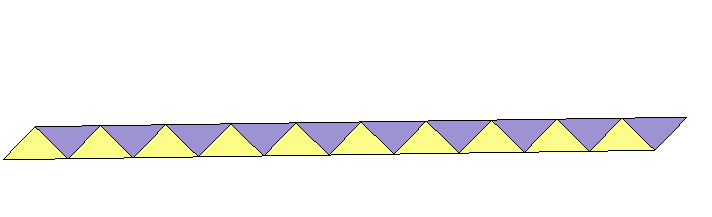}
        \caption{Classical Miura-ori strip}
        \label{fig:miura_ori_strip}
    \end{subfigure}
    \caption{Classical Miura-ori and its realization in 3D}
    \label{fig:miura}
\end{figure}
To generate more Miura-ori type domains, ideally we can simply apply a conformal map on the domain. Notice that, in the continuous case, the domain obtained by compositing a flat-foldable configuration with a conformal map remains flat-foldable (satisfying the alternating Beltrami equation with the same coefficients as before), because of the angle preserving property. Such a composition can create triangles at different scales.

However, because of the discreteness, the angles is preserved only if the map is a uniform scaling plus rigid motion. Indeed, this follows from the our assertion on rank of the system matrix. Fortunately, applying a conformal map usually only yield a small and structured perturbation to the Beltrami equation, and the new Miura-ori domains can still be created via several iteration of the foldings and unfoldings, in light of the reinforcement iteration we proposed. For example, a new Miura-ori pattern in Figure \ref{fig:miura2} is created via this method, with the choice of  (in this case we just made any convenient choice)
\[
\Phi(z) = 10 + 0.1z + 0.4z^2.
\]

\begin{figure}
    \centering
    \begin{subfigure}[b]{0.33\textwidth}
        \includegraphics[width=\textwidth]{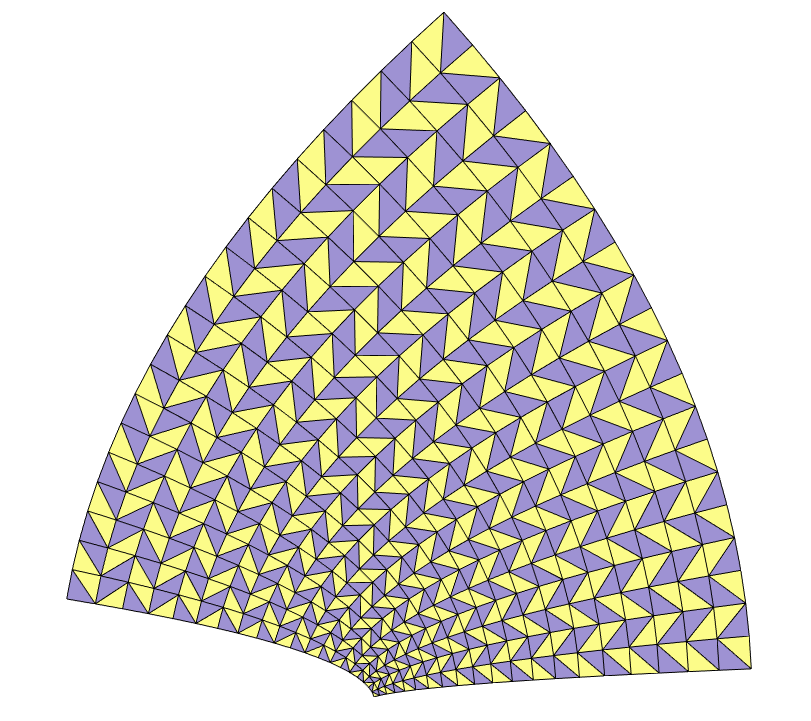}
        \caption{Naive composition with $\Phi$}
        \label{fig:miura_ori_unfolded2}
    \end{subfigure}
    \hspace{10pt}
    \begin{subfigure}[b]{0.33\textwidth}
        \includegraphics[width=\textwidth]{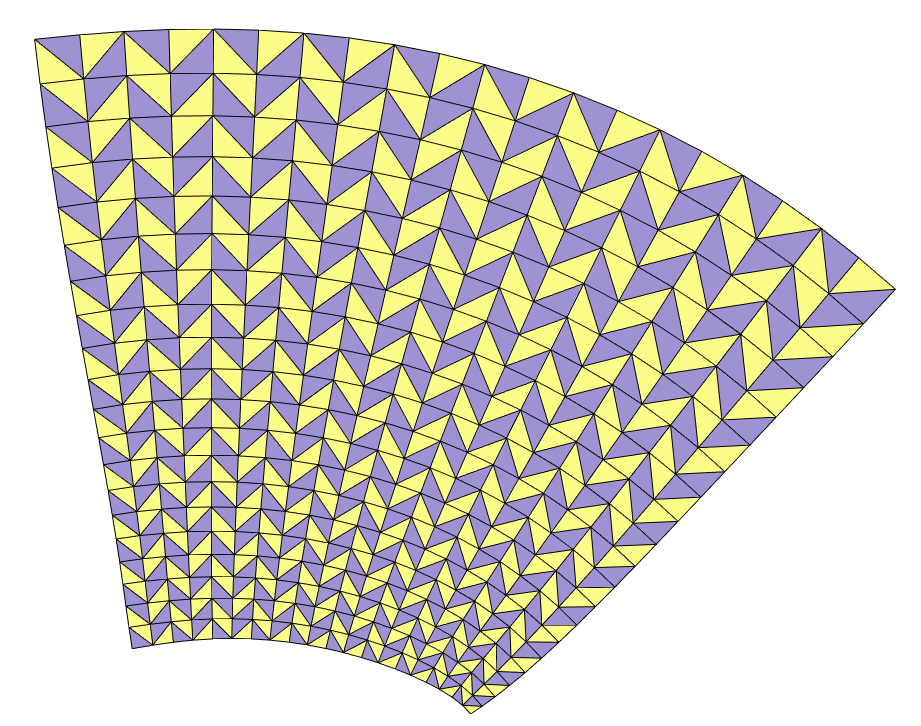}
        \caption{A new Miura-ori pattern after folding-unfolding iterations (maximal distortion $ = 3\times 10^{-4}$ compared to the folded state (c))}
        \label{fig:miura_ori_unfolded22}
    \end{subfigure}
    \begin{subfigure}[b]{0.45\textwidth}
        \includegraphics[width=\textwidth]{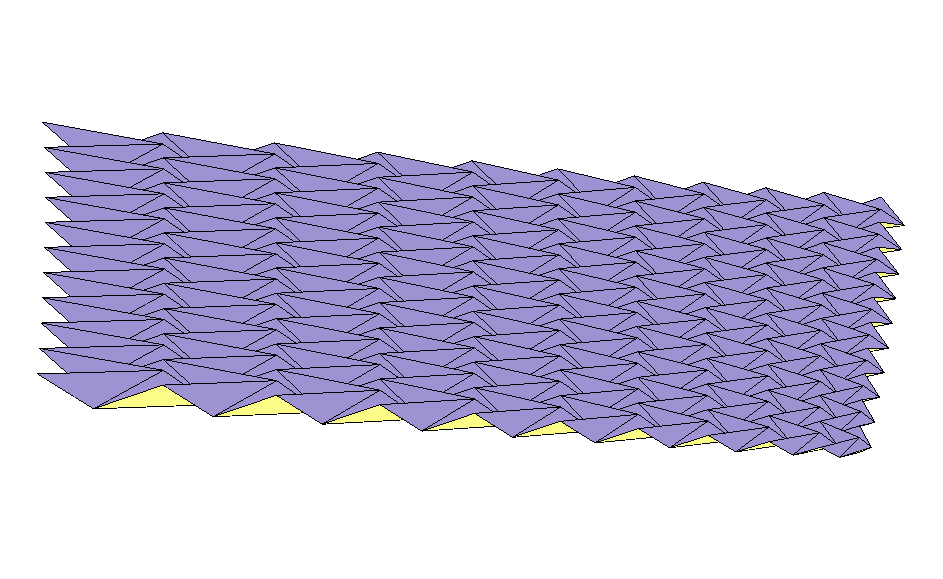}
        \caption{Realization of the Miura-ori in (b)}
        \label{fig:miura_ori_folded2}
    \end{subfigure}
    \hspace{10pt}
    \begin{subfigure}[b]{0.39\textwidth}
        \includegraphics[width=\textwidth]{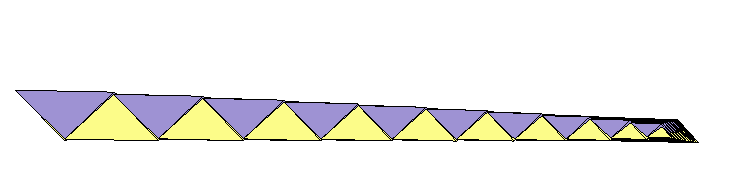}
        \caption{New Miura-ori strips: note that there is a stack of strips of varing sizes}
        \label{fig:miura_ori_strip2}
    \end{subfigure}
    \caption{A new Miura-ori pattern by composition with $\Phi$ and its realization in 3D. Here, maximal distortion is defined by $\max\{\max\{|\mu_T|\}_{T\in \Omega^{+}}, \max\{1/|\mu_T|\}_{T\in \Omega^{-}} \}$}
    \label{fig:miura2}
\end{figure}

 Different from approach of Dudte {\it et al.} \cite{dudte2016programming}, the surface we obtain is flat-foldable by design. Given the rich family of conformal maps, it will be particularly interesting to study the new family of Miura-ori patterns with the aid of our algorithm. The study of different patterns' curvature approximation capacities is also a exciting future direction. We envisage a ``conformal geometric processing" approach to the modelling of Miura-ori. Under such an approach researchers can efficiently design the pattern with a simple set of CAD tools. Mathematical understanding of this problem will definitely benefit such a ``bottom-up" approach to material design with flat-foldable structures.

For a preliminary example, we can simulate and study the deformation of the Miura-ori in 3D with our solutions. Starting from the flat-folded state of the surface, one can apply the classical geometric editing methods such as as-rigid-as-possible \cite{sorkine2007rigid}. An example of such a deformation with user-defined position constraints is shown in Figure \ref{fig:miura2_arap}. 

\begin{figure}
    \centering
    \includegraphics[width=0.5\textwidth]{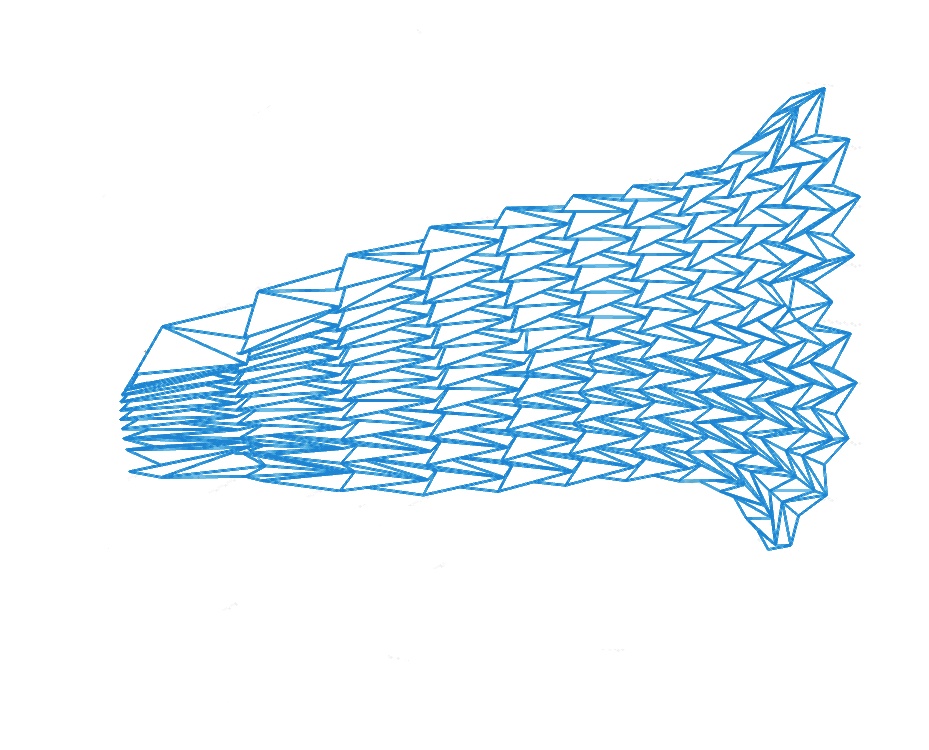}
    \caption{A rigid deformation of the Miura-ori in Figure \ref{fig:miura2}}
    \label{fig:miura2_arap}
\end{figure}
\subsection{Almost rigid folding with application to fold-texture generation, fold sculpting and fold in-painting}
As one of the immediate applications, we can consider a folding transformation on the texture space to create synthetic fold-like textures, prior to applying the texture map.  This can be cheap to do if high quality physical simulation and rendering is not available. In Figure \ref{fig: 3dfold_texture} we explore such a possibility of user-designed fold-like texture generation. 

One of the fundamental steps in texturing a 3D surface is to find the parameterization (or the texture map) $f : S \to \Omega \in \mathbb{R}^2$. In particular, UV map is one of the major types of parameterization techniques in various software packages, which works well if the 3D model is created from polygon meshes. The above technique can be very useful in the interactive user design, where the user directly operates on the target mesh, and the input is transformed to the texture domain via the UV-map, to create desirable fold-like texture on the target mesh. It is also possible to incorporate proper shading effect on the transformed texture directly, making the texture look more realistic. We have implemented such a fold-like texturing method using a 3D T-shirt model, shown in Figure \ref{fig:3dtexture}. Note that the mesh is not deformed at all.

\begin{figure}
    \centering
    \begin{subfigure}[b]{0.35\textwidth}
    \includegraphics[width=\textwidth]{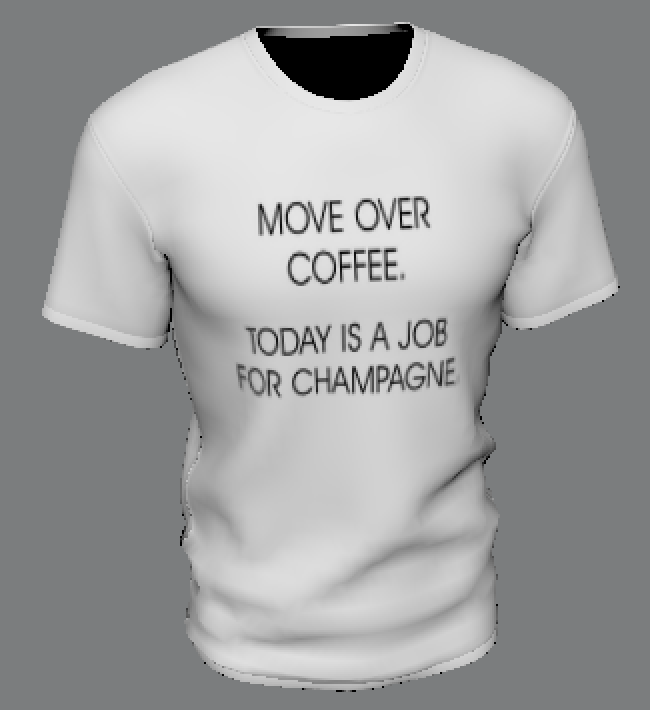}
    \caption{Before fold-texturing}
    \label{fig:original}
    \end{subfigure}
    \begin{subfigure}[b]{0.36\textwidth}
    \includegraphics[width=\textwidth]{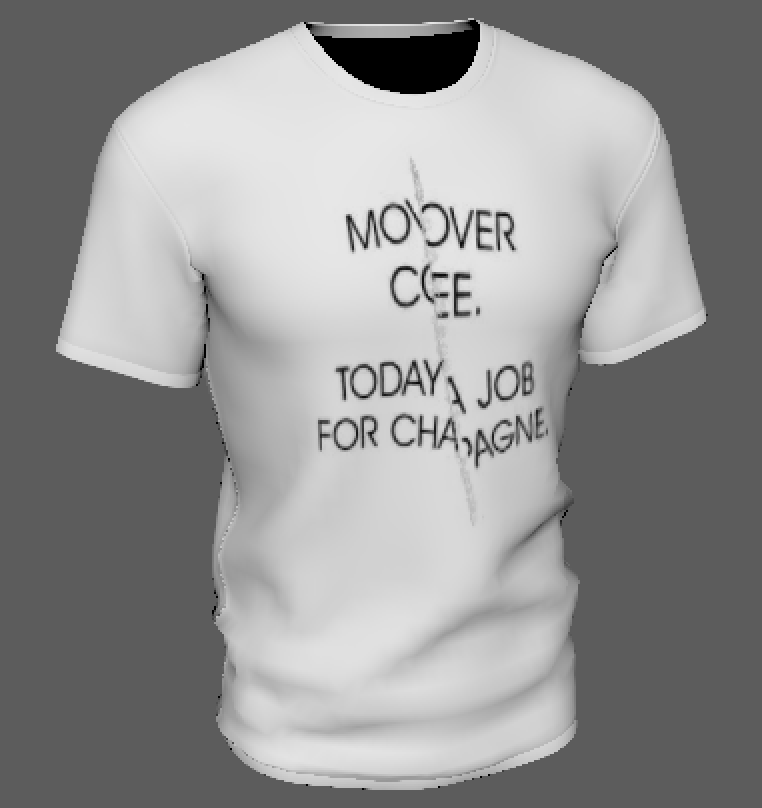}
    \caption{After fold-texturing}
    \label{fig:3dtexture}
    \end{subfigure}

    \caption{Folding effect texturing on a 3D model. Note that the mesh model is not deformed.}
    \label{fig: 3dfold_texture}
\end{figure}

We can also apply the folding technique directly to the 3D meshes, as an application we would like to call {\it fold sculpting}. To illustrate this, we select a patch from the T-shirt model, as shown in Figure \ref{fig:tshirt}. We applied the folding operation to a suitable parametrization of the patch, which can obtained easily via, for example, projection or a least square conformal parametrization \cite{levy2002least}, and then glue it back to the T-shirt model. Note that our algorithm produces sharp edges. This can be mitigated by some standard smoothing operation in various mesh editing software. Figure \ref{fig: maya_smooth2} and \ref{fig: maya_smooth} show the results after appropriate smoothing, where we used the software Maya\footnotemark to the smoothing and rendering tasks. Note that such folding is not easily obtained by pure handcraft, since one part of the cloth actually folds over and covers some other part of the cloth. 
\footnotetext{A software of the Autodesk Inc. See \url{https://www.autodesk.com.hk/products/maya/overview}. The results are generated under the student license  obtained by the first author. }
\begin{figure}
    \centering
    \begin{subfigure}[b]{0.30\textwidth}
    \includegraphics[width=\textwidth]{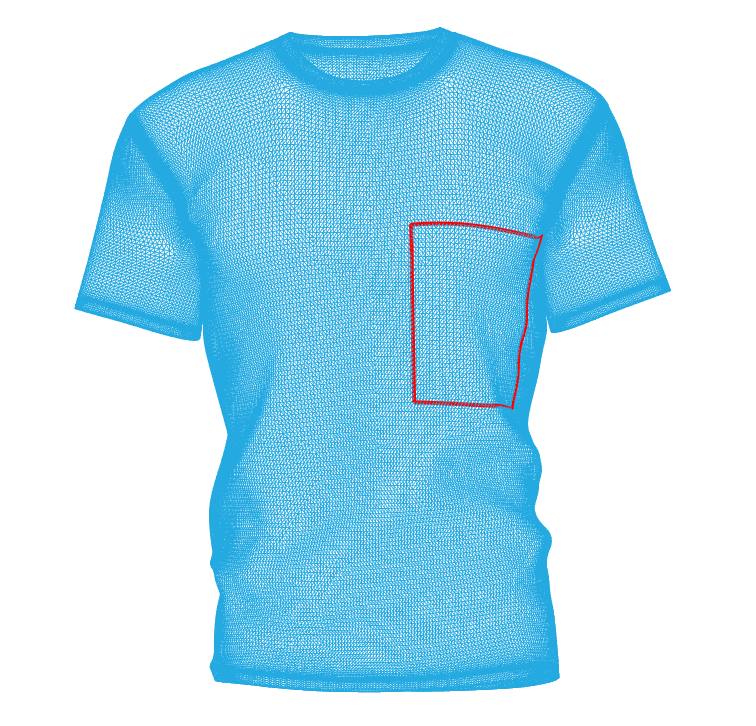}
    \caption{The T-shirt model}
    \label{fig:tshirt}
    \end{subfigure}
    \begin{subfigure}[b]{0.34\textwidth}
    \includegraphics[width=\textwidth]{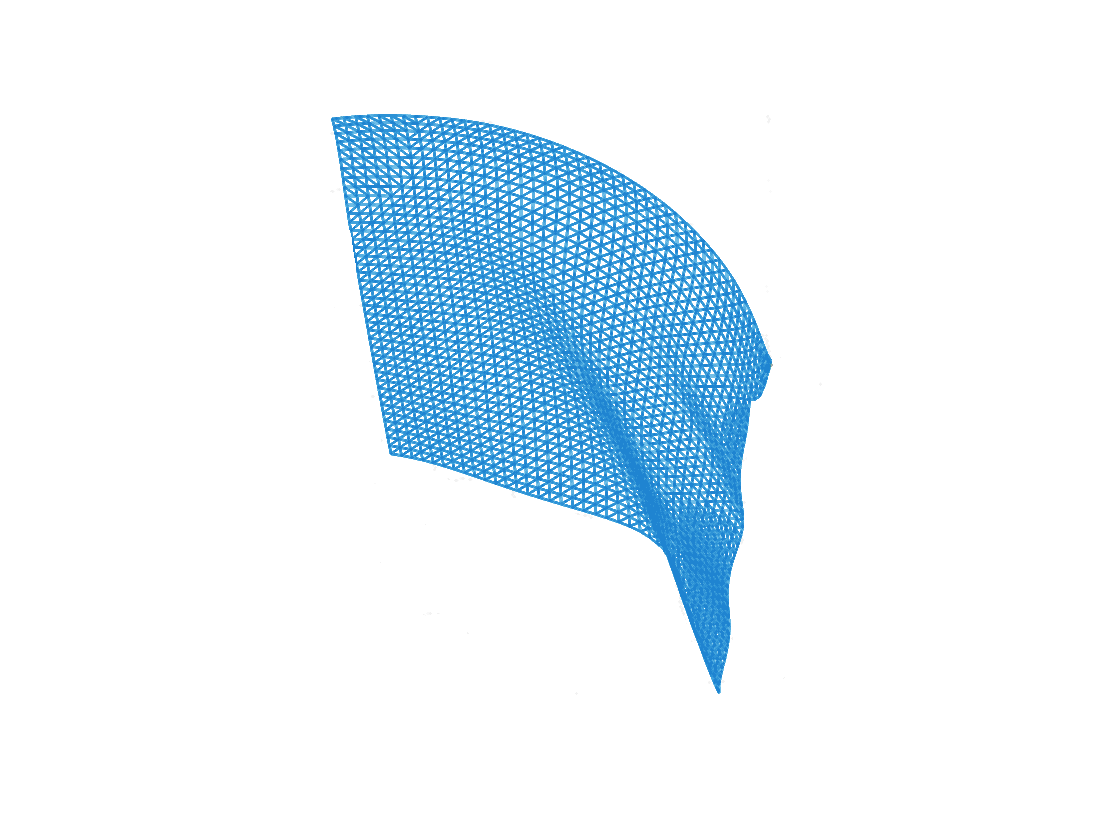}
    \caption{The original patch}
    \label{fig:ori_patch}
    \end{subfigure}
    \begin{subfigure}[b]{0.34\textwidth}
    \includegraphics[width=\textwidth]{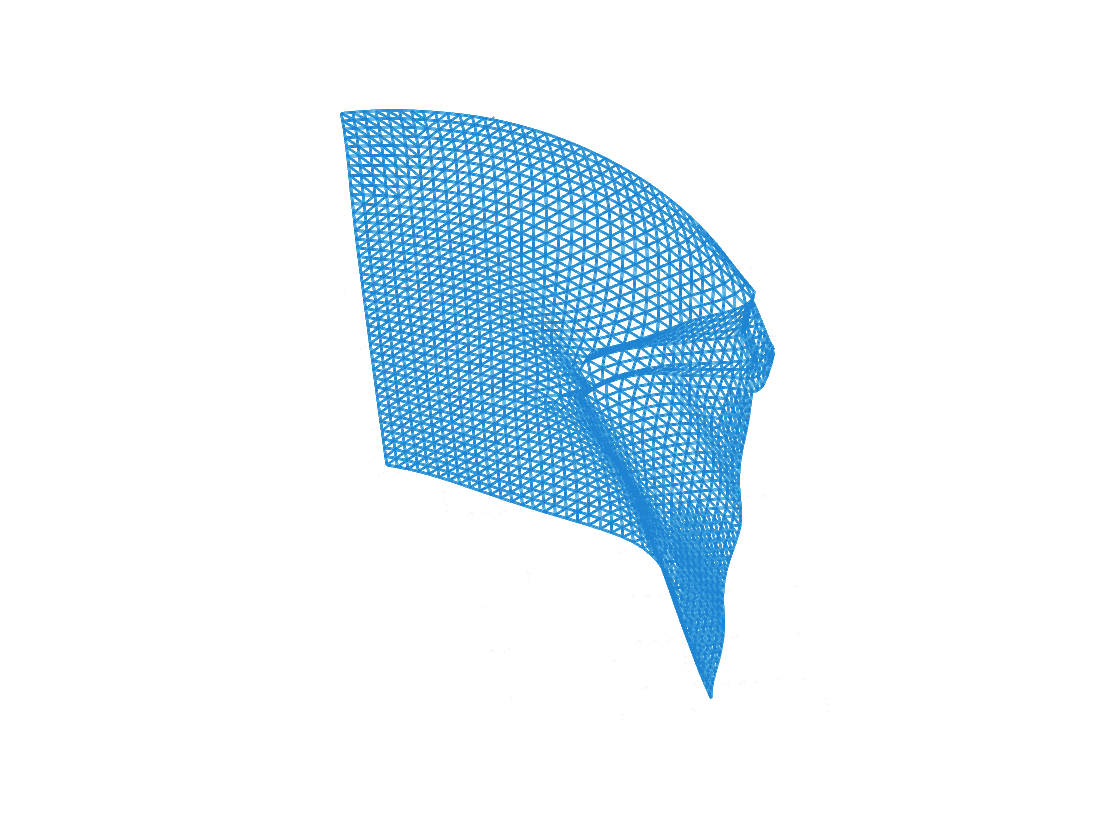}
    \caption{The deformed patch}
    \label{fig:def_patch}
    \end{subfigure}
    \caption{Patch-wise fold sculpting: the region inside the red contour is the patch selected, appropriate alternating Beltrami equation is then solved in the patch domain to obtain the desired folding effect.}
    \label{fig: 3dfold}
\end{figure}

\begin{figure}
    \centering
    \begin{subfigure}[b]{0.31\textwidth}
    \includegraphics[width=\textwidth]{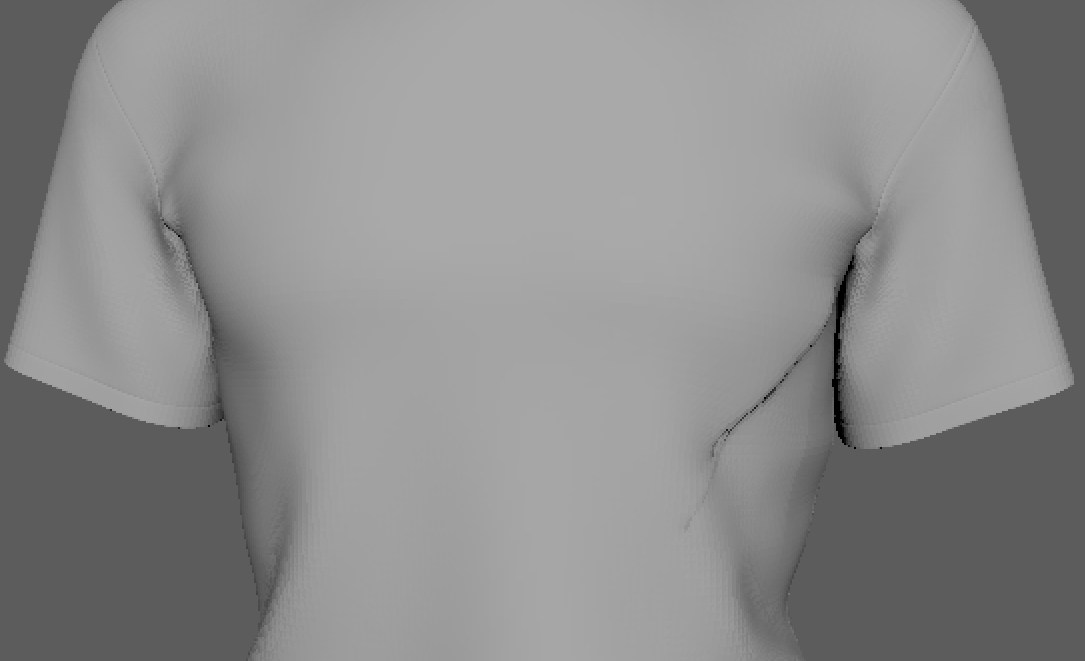}
    \end{subfigure}
    \begin{subfigure}[b]{0.32\textwidth}
    \includegraphics[width=\textwidth]{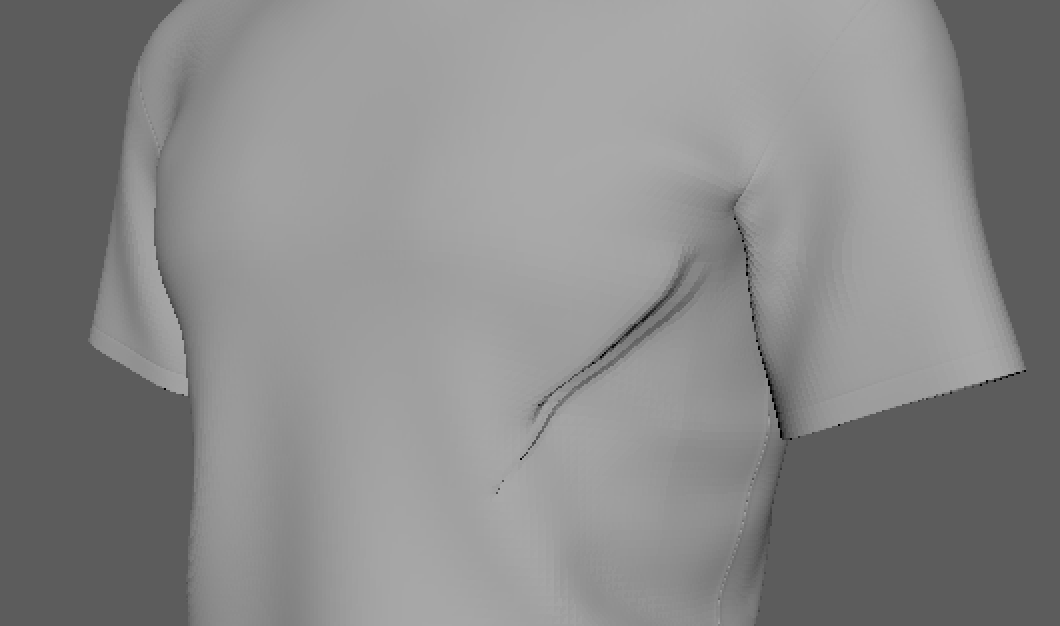}
    \end{subfigure}
    \begin{subfigure}[b]{0.33\textwidth}
    \includegraphics[width=\textwidth]{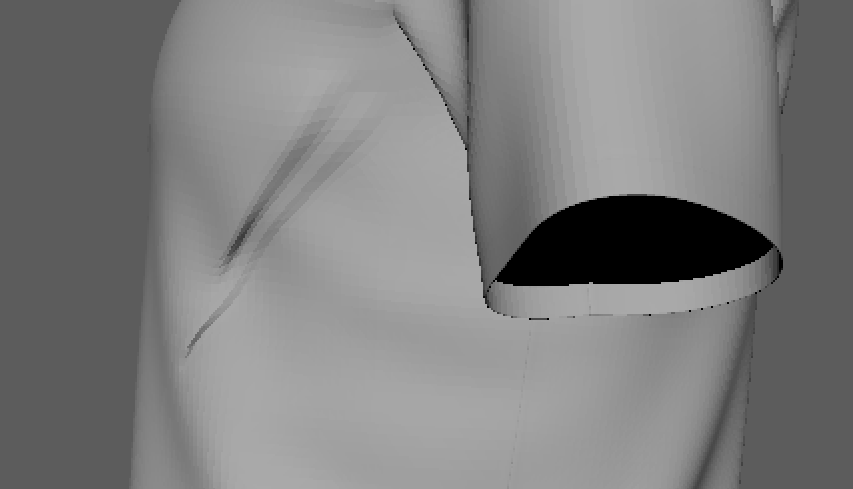}
    \end{subfigure}
    \caption{Results of fold sculpting on the T-shirt model after appropriate smoothing: two short folds are sculpted on the right.}
    \label{fig: maya_smooth2}
\end{figure}

\begin{figure}
    \centering
    \begin{subfigure}[b]{0.35\textwidth}
    \includegraphics[width=\textwidth]{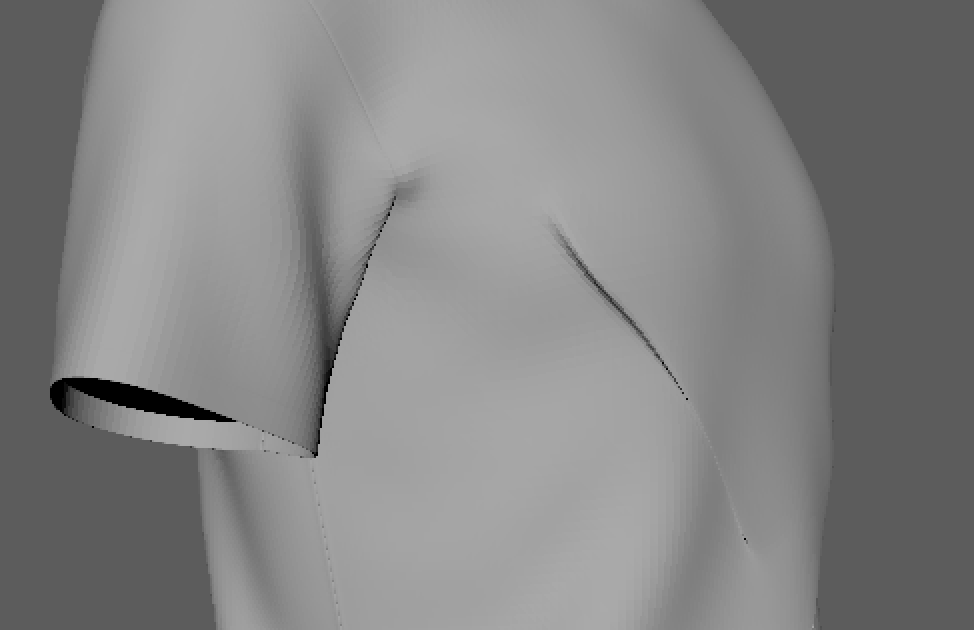}
    \end{subfigure}
    \begin{subfigure}[b]{0.35\textwidth}
    \includegraphics[width=\textwidth]{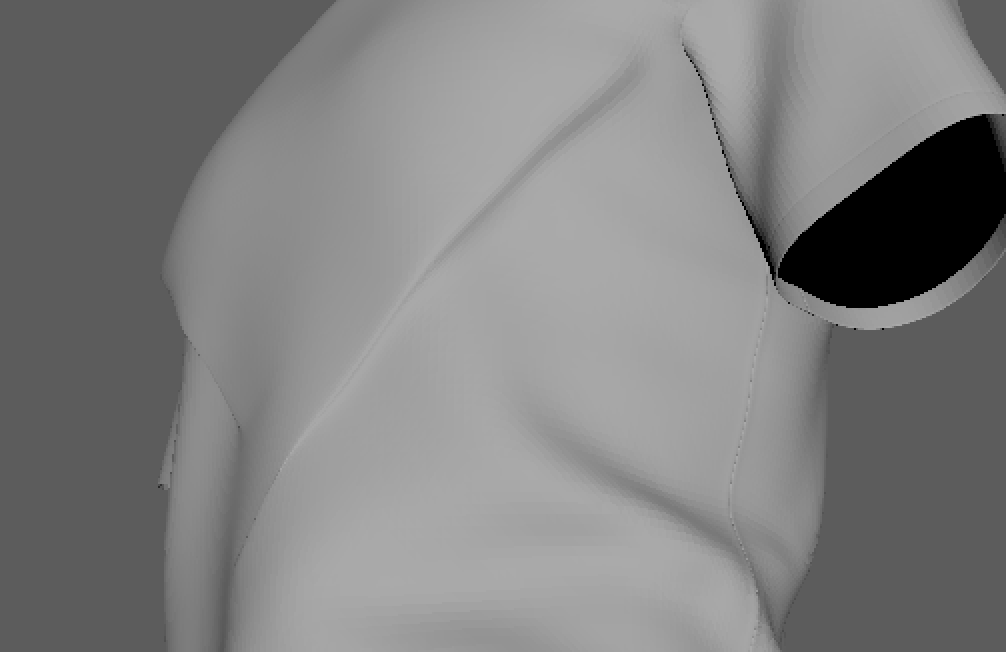}
    \end{subfigure}
    \caption{Results of fold sculpting on the T-shirt model: two long folds are sculpted on the left and right.}
    \label{fig: maya_smooth}
\end{figure}

The technique can also be applied after the acquisition of a folded surface using laser scans, where the folded part introduces self-occlusions and the folding is usually diminished or destroyed after applying the watertight operation. To preserve the folding details from the scans directly, we can mark the folding part that we want to preserve in the raw acquisition. By taking a patch like before and mapping it into the plane, we can solve a proper alternating Beltrami equation to obtain the desired folding effect. The folded patch can then be mapped back to the raw acquisition. To illustrate this, we have done a synthetic experiment using the above approach. 
\begin{figure}
    \centering
    \begin{subfigure}[b]{0.32\textwidth}
    \includegraphics[width=\textwidth]{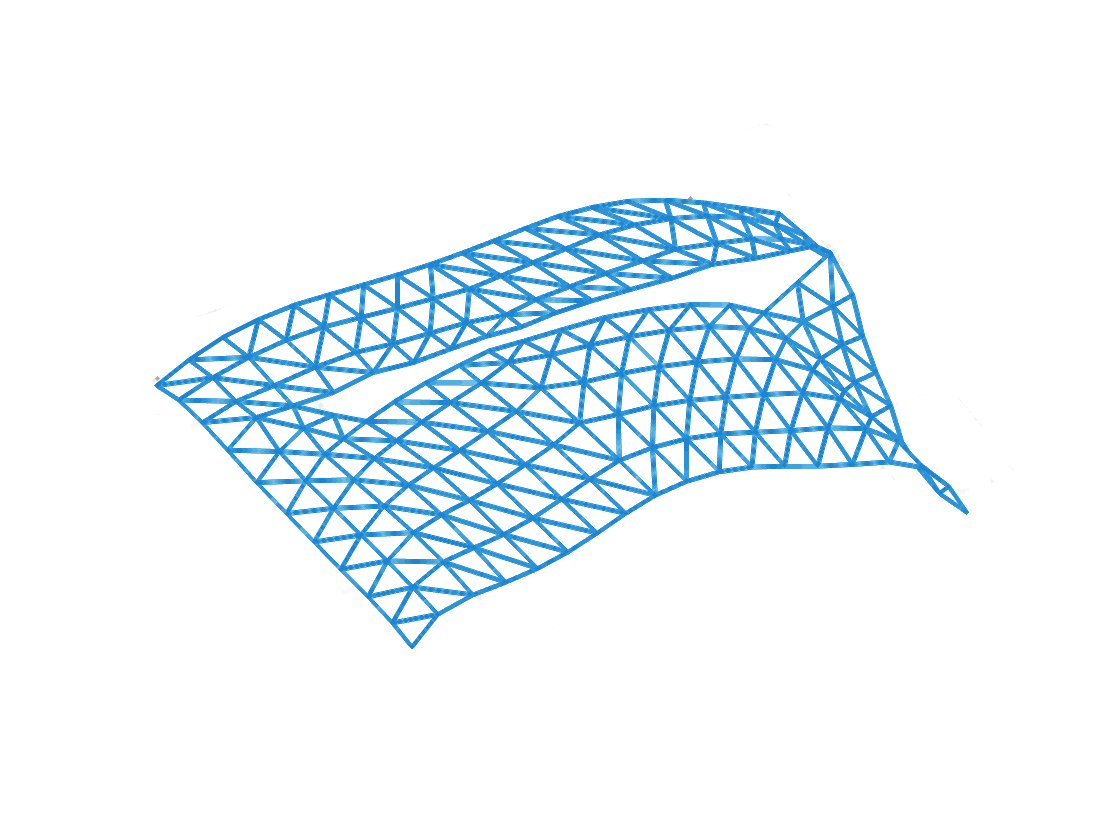}
    \end{subfigure}
    \begin{subfigure}[b]{0.33\textwidth}
    \includegraphics[width=\textwidth]{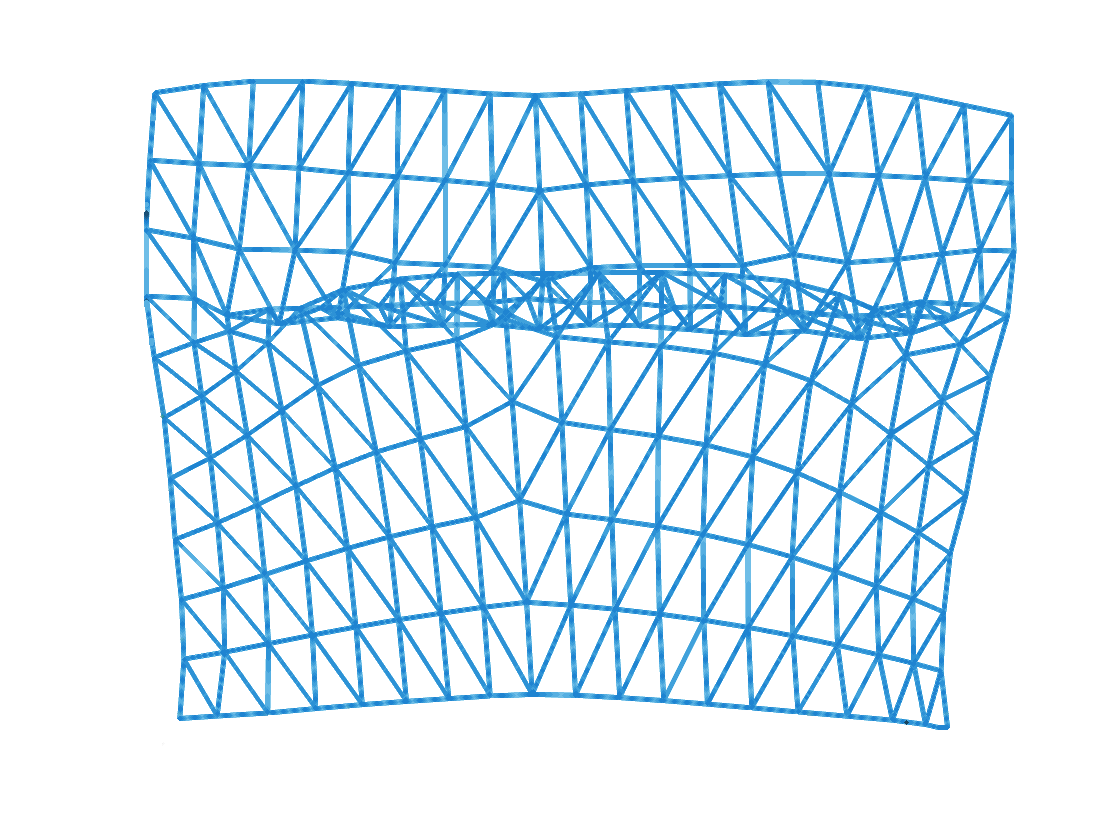}
    \end{subfigure}
    \begin{subfigure}[b]{0.33\textwidth}
    \includegraphics[width=\textwidth]{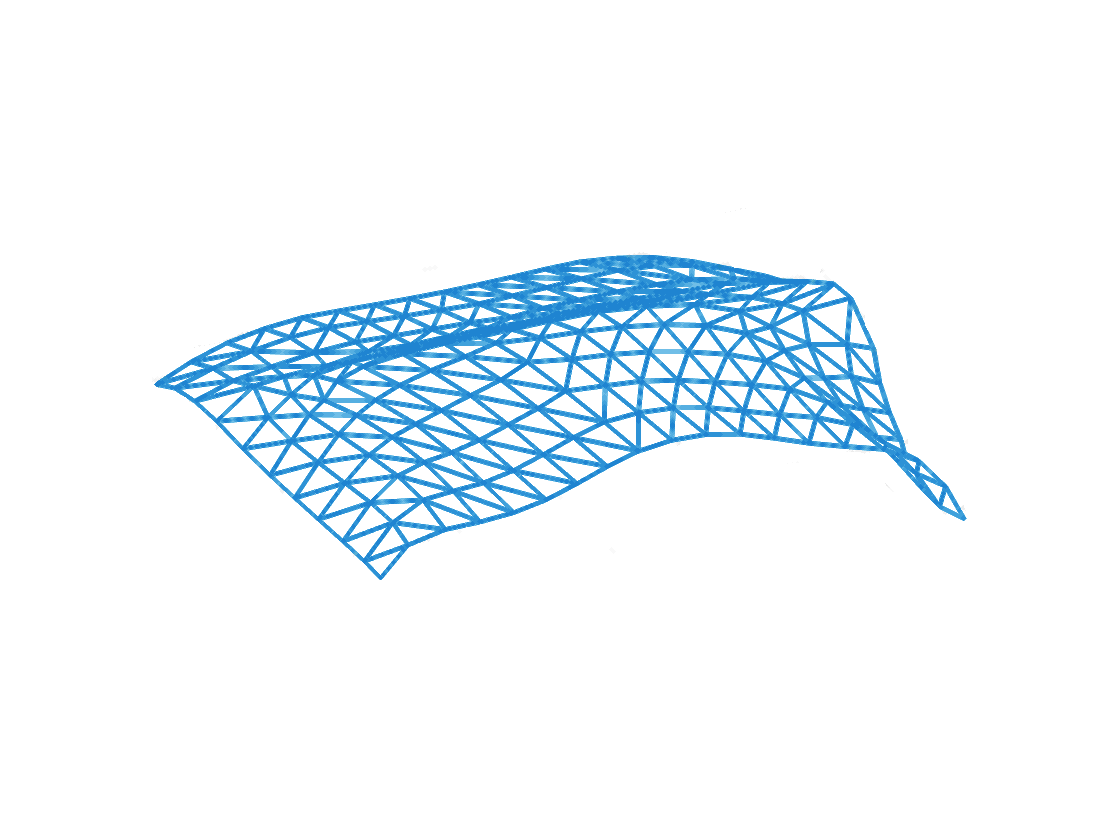}
    \end{subfigure}
    \caption{Patch-wise fold in-painting: from the raw acquisition of a self-occluded surface, we map the occluded, holed surface to the plane and fill the hole; we then apply a suitable folding operation to reproduce the fold that was not captured during acquisition.}
    \label{fig: paint patch}
\end{figure}
We begin with an incomplete acquisition of a shirt model and a dress model, such as the one on the left in Figure \ref{fig: paint fold}. As demonstrated in Figure \ref{fig: paint patch}, the patch with holes is first map to the plane and subsequently filled. A suitable folding operation is then applied to the patch to produce a plausible fold geometry given the acquisition data. The reconstruction is shown in Figure \ref{fig: paint fold} on the right. 
\begin{figure}
    \centering
    \begin{subfigure}[b]{0.22\textwidth}
    \includegraphics[width=\textwidth]{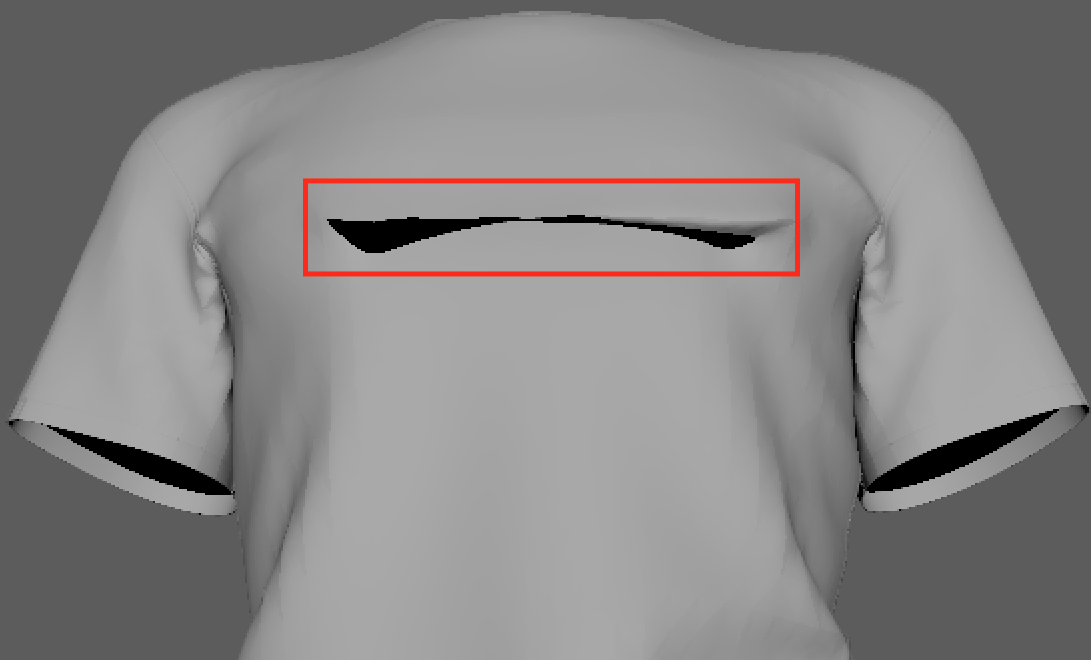}
    \end{subfigure}
     \begin{subfigure}[b]{0.22\textwidth}
    \includegraphics[width=\textwidth]{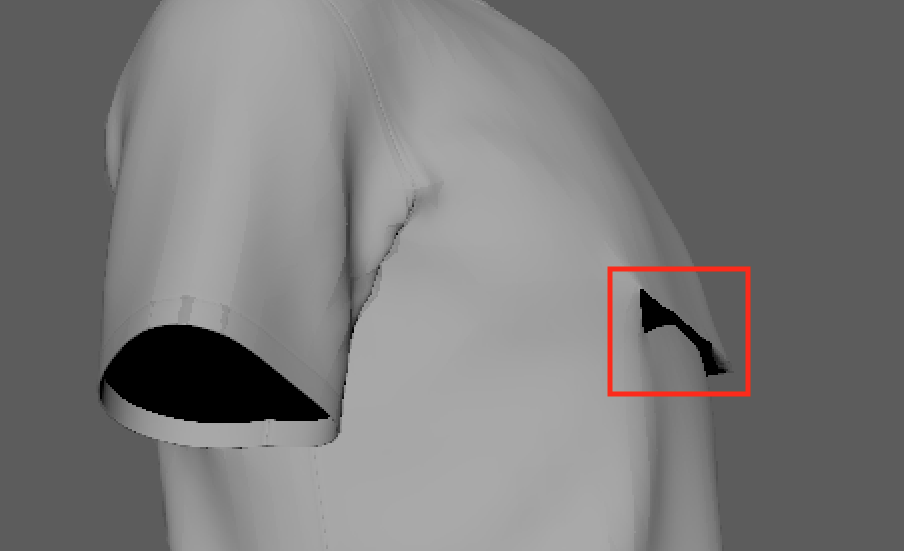}
    \end{subfigure}
    \begin{subfigure}[b]{0.22\textwidth}
    \includegraphics[width=\textwidth]{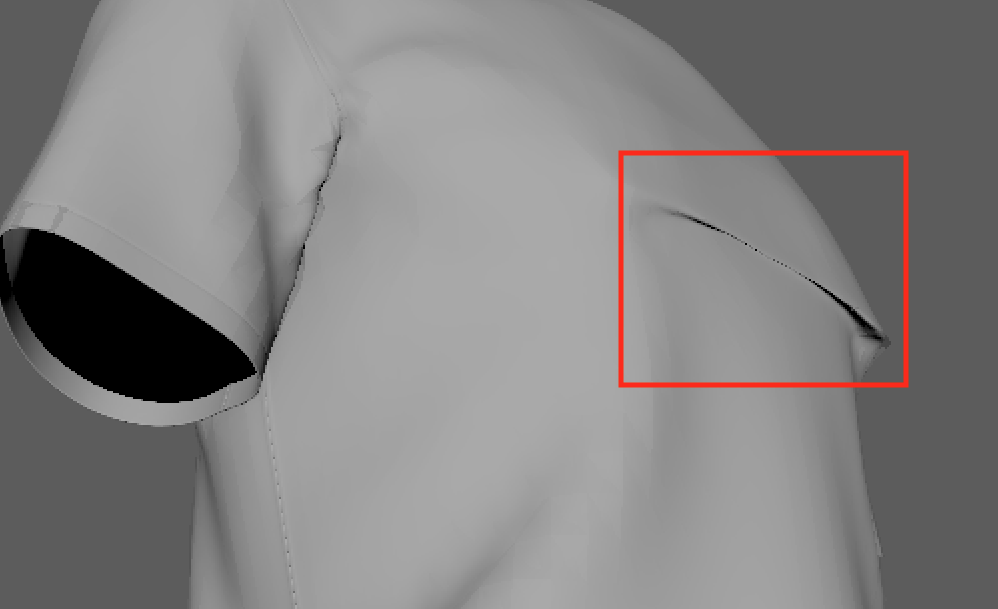}
    \end{subfigure}
    \begin{subfigure}[b]{0.22\textwidth}
    \includegraphics[width=\textwidth]{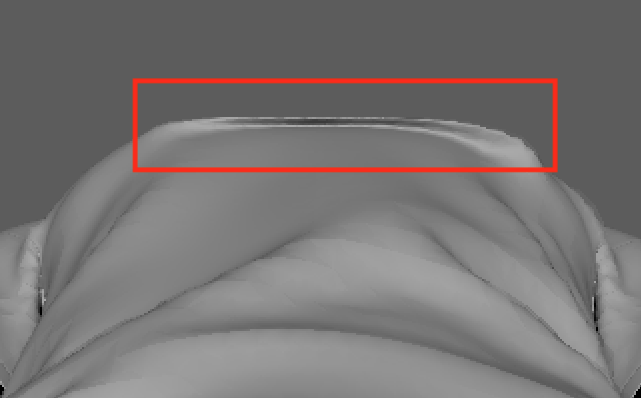}
    \end{subfigure}
    \begin{subfigure}[b]{0.24\textwidth}
    \includegraphics[width=\textwidth]{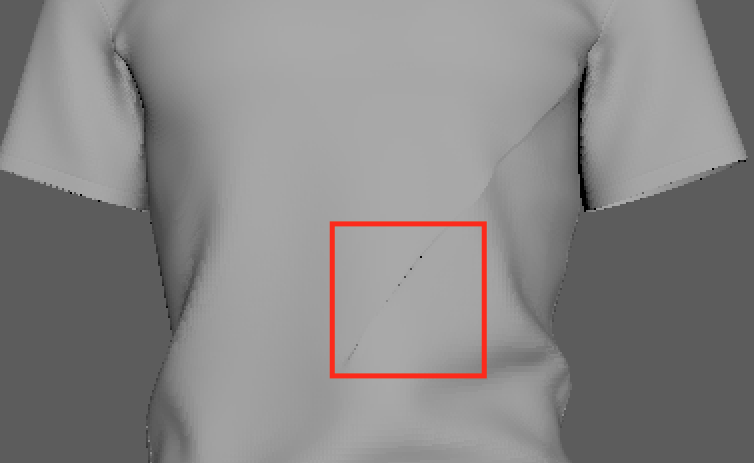}
    \end{subfigure}
    \begin{subfigure}[b]{0.24\textwidth}
    \includegraphics[width=\textwidth]{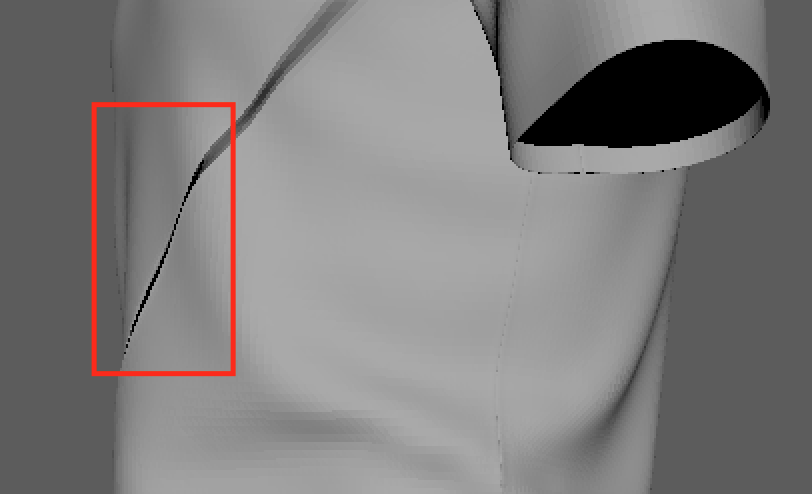}
    \end{subfigure}
    \begin{subfigure}[b]{0.24\textwidth}
    \includegraphics[width=\textwidth]{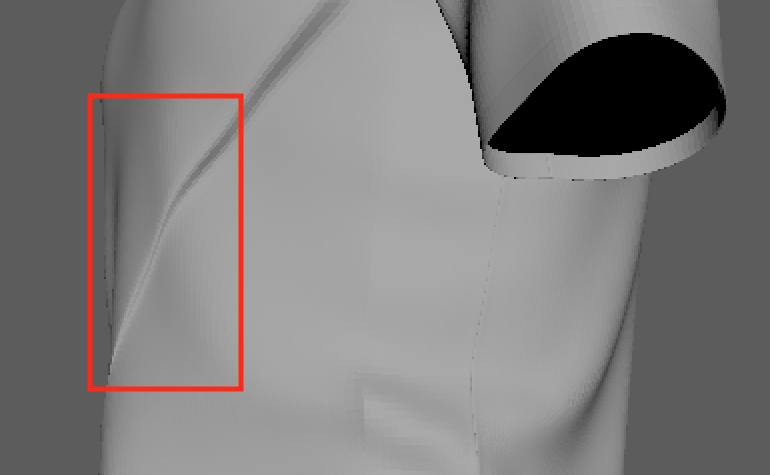}
    \end{subfigure}\\
    
    \begin{subfigure}[b]{0.24\textwidth}
    \includegraphics[width=\textwidth]{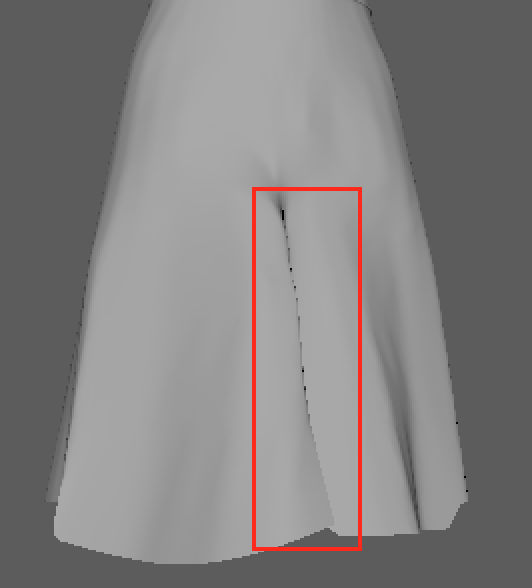}
    \end{subfigure}
    \begin{subfigure}[b]{0.235\textwidth}
    \includegraphics[width=\textwidth]{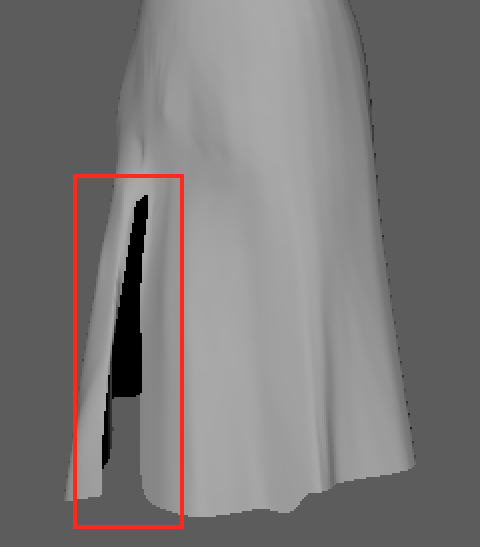}
    \end{subfigure}
    \begin{subfigure}[b]{0.235\textwidth}
    \includegraphics[width=\textwidth]{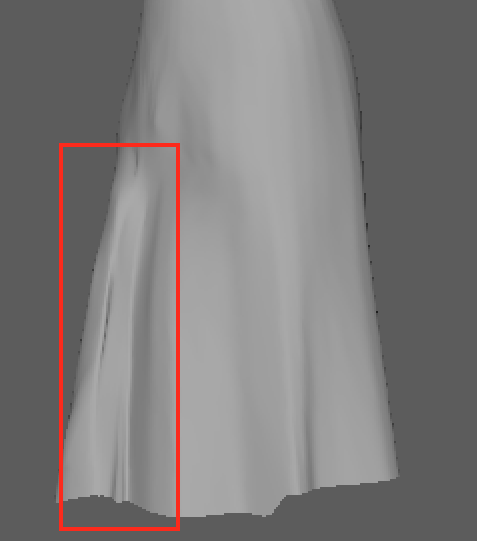}
    \end{subfigure}
    \caption{Results of fold in-painting. For each row, the two on the left are the surface with holes from acquisition due to self occlusion. The results of in-painting are on the right. The corresponding holed and in-painted regions are highlighted inside the red boxes.}
    \label{fig: paint fold}
\end{figure}

\subsection{Self-occlusion reasoning of flat-foldable surfaces and its application to restoration of folded images}
Given a single perspective of a folded surface, for example, shown in figure \ref{fig:experiemnt234-fold}, we can use the proposed reinforcement iteration to unfold the surface, thus enabling us to identify the self-occluded region in the unfolded domains, shown in Figure \ref{fig:mask_regions}.  

\begin{figure}
    \centering
    \begin{subfigure}[b]{0.22\textwidth}
        \includegraphics[width=\textwidth]{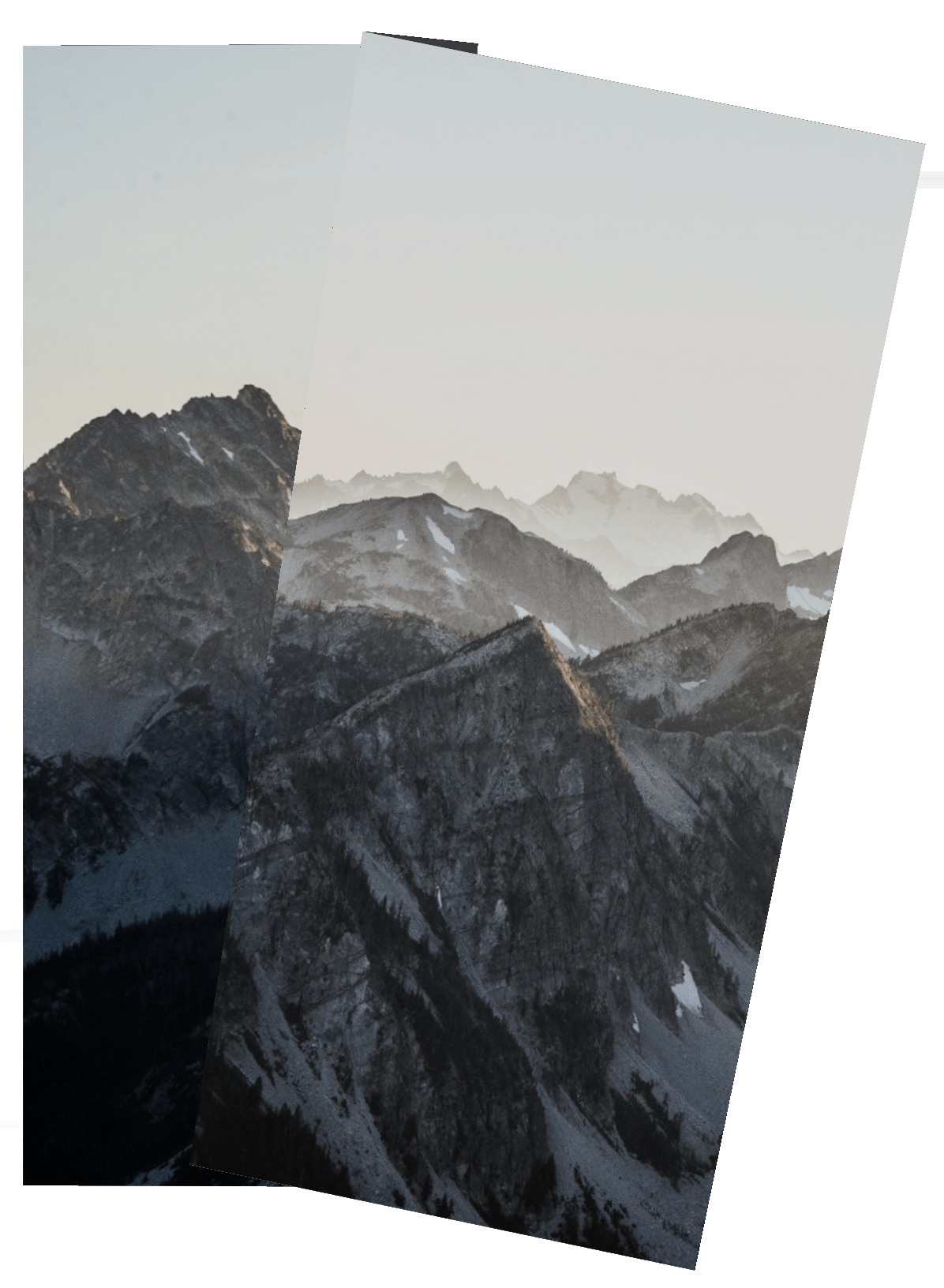}
        \caption{"Mountain"}
        \label{fig:onefold_a}
    \end{subfigure}
    \begin{subfigure}[b]{0.27\textwidth}
        \includegraphics[width=\textwidth]{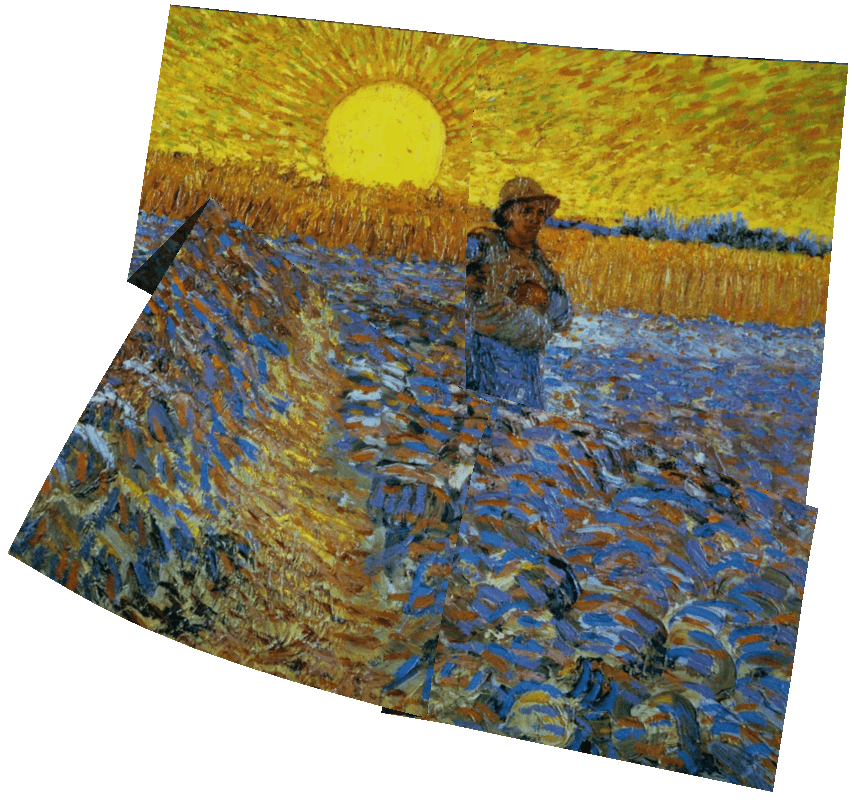}
        \caption{"The Sower"}
        \label{fig:twofold_a}
    \end{subfigure}
    \begin{subfigure}[b]{0.18\textwidth}
        \includegraphics[width=\textwidth]{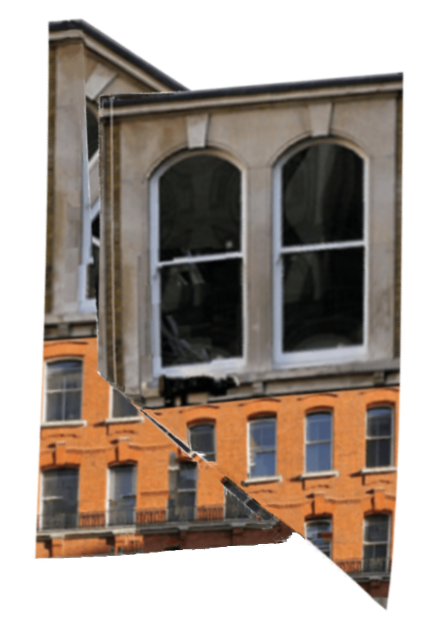}
        \caption{"Building"}
        \label{fig:3fold-a}
    \end{subfigure}
    
    \caption{1-fold, 2-fold and cusp-fold examples}\label{fig:experiemnt234-fold}
\end{figure}

\begin{figure}
    \centering
    \begin{subfigure}[b]{0.23\textwidth}
        \includegraphics[width=\textwidth]{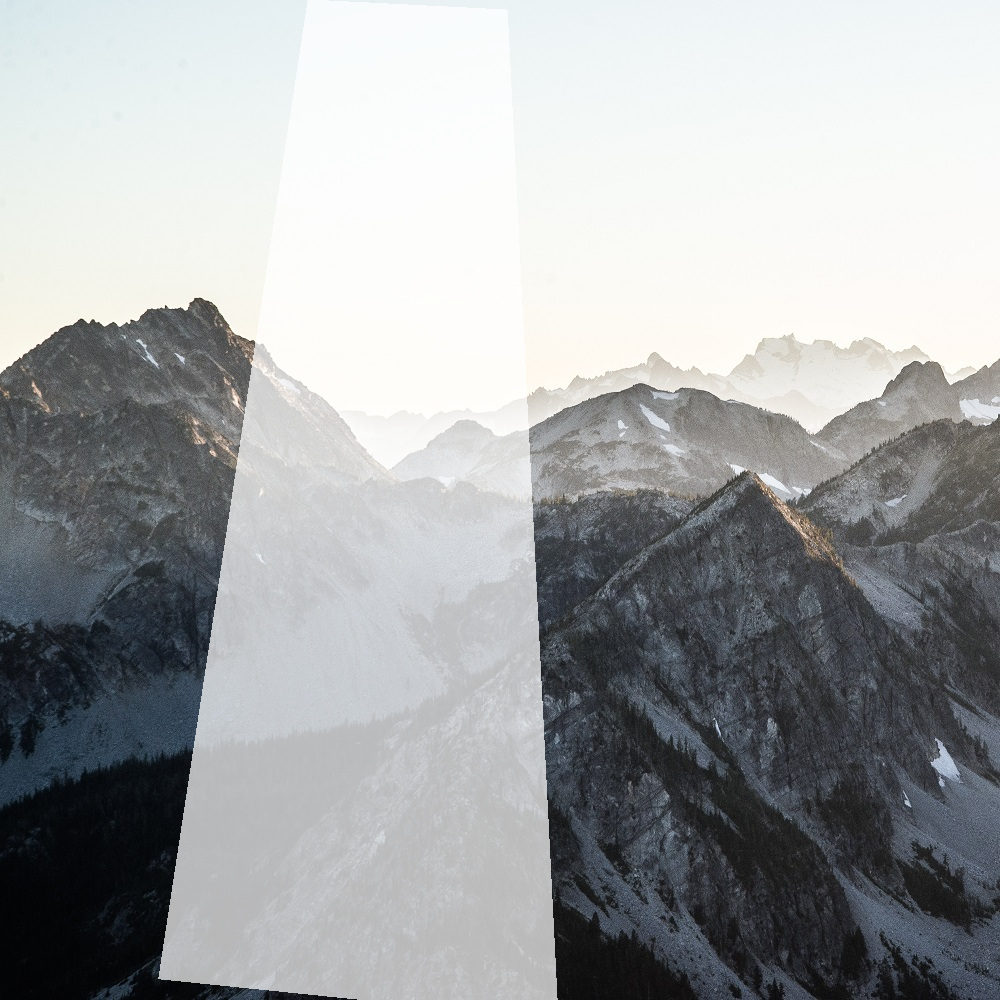}
        \caption{1-fold}
        \label{fig:ex1_mask}
    \end{subfigure}
    
    \begin{subfigure}[b]{0.23\textwidth}
        \includegraphics[width=\textwidth]{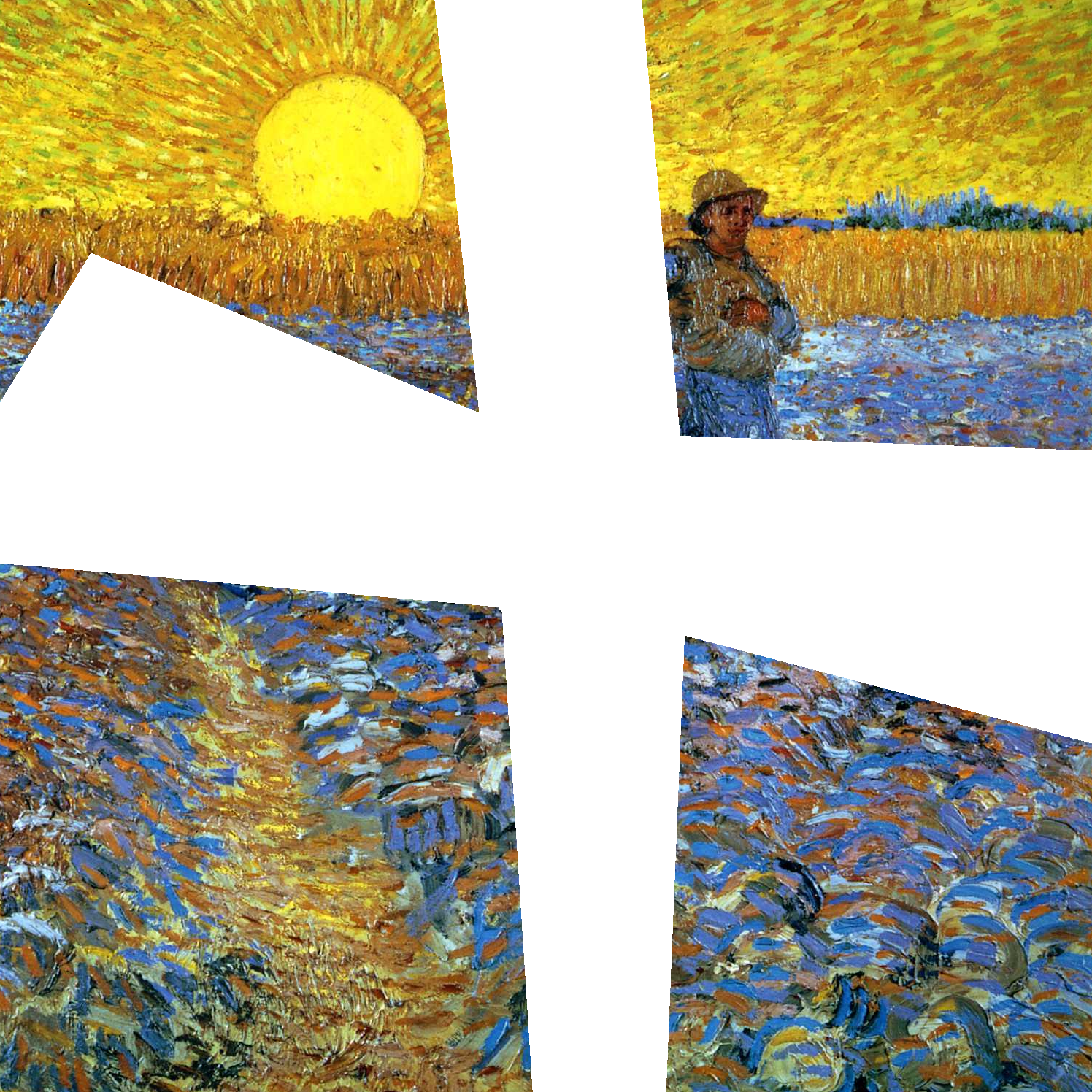}
        \caption{2-fold}
        \label{fig:ex2_mask}
    \end{subfigure}
    \begin{subfigure}[b]{0.23\textwidth}
        \includegraphics[width=\textwidth]{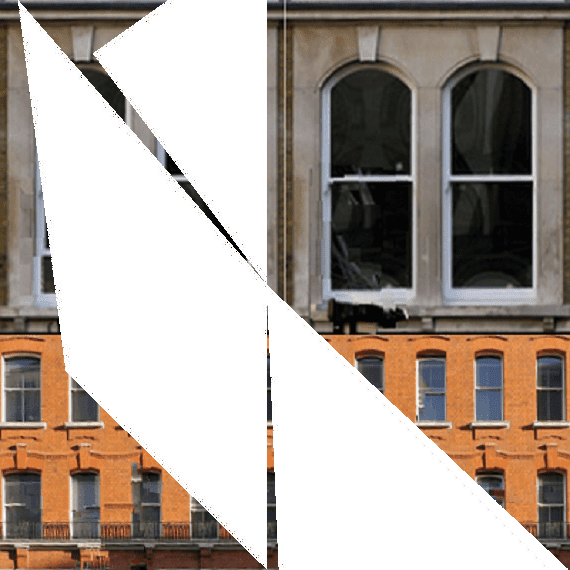}
        \caption{cusp-fold}
        \label{fig:ex3_mask}
    \end{subfigure}
    \caption{Occluded regions of various folded paper examples}\label{fig:mask_regions}
\end{figure}

Given a folded image, the task of restoring the image involves unfolding the image and in-painting the missing parts beneath the folded region. The performance of the final restoration results obviously depends on the realization of the texture synthesis. However, it is worth noting that the unfolding result may also drastically affects the in-painting result since many in-painting methods such as the diffusion-based \cite{bertalmio2000image, tschumperle2005vector}, exemplar-based \cite{daisy2013fast, barnes2009patchmatch, barnes2010generalized} assume the full knowledge of the computational domain (i.e., the image domain). To alleviate the difficulty arose from the incomplete knowledge of the image domain, we can employ the proposed unfolding algorithm to retrieve the geometric information of the folded subdomain with the given partial geometric information. Once we restore the intrinsic image domain consistent with the partial geometric information, well-developed in-painting techniques can then be employed correctly and provide satisfactory in-painting results.

    

\begin{figure}
    \centering
    \begin{subfigure}[b]{0.22\textwidth}
        \includegraphics[width=\textwidth]{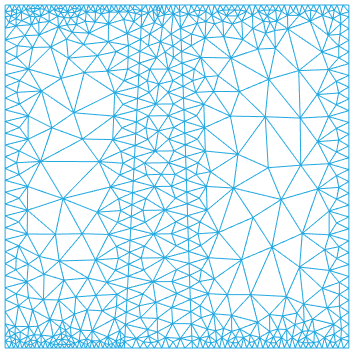}
        \caption{Ground truth}
        \label{fig:onefoldresult_a}
    \end{subfigure}
    \begin{subfigure}[b]{0.22\textwidth}
        \includegraphics[width=\textwidth]{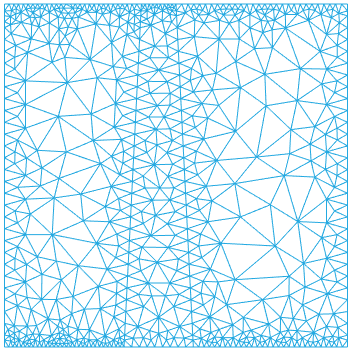}
        \caption{Unfolded mesh}
        \label{fig:onefoldresult_b}
    \end{subfigure}
    \begin{subfigure}[b]{0.20\textwidth}
        \includegraphics[width=\textwidth]{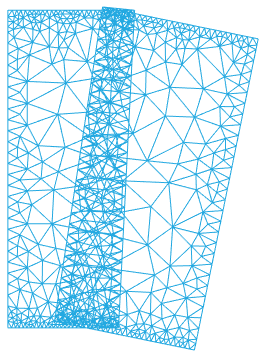}
        \caption{Parametrized}
        \label{fig:onefold_b}
    \end{subfigure}
    \begin{subfigure}[b]{0.22\textwidth}
        \includegraphics[width=\textwidth]{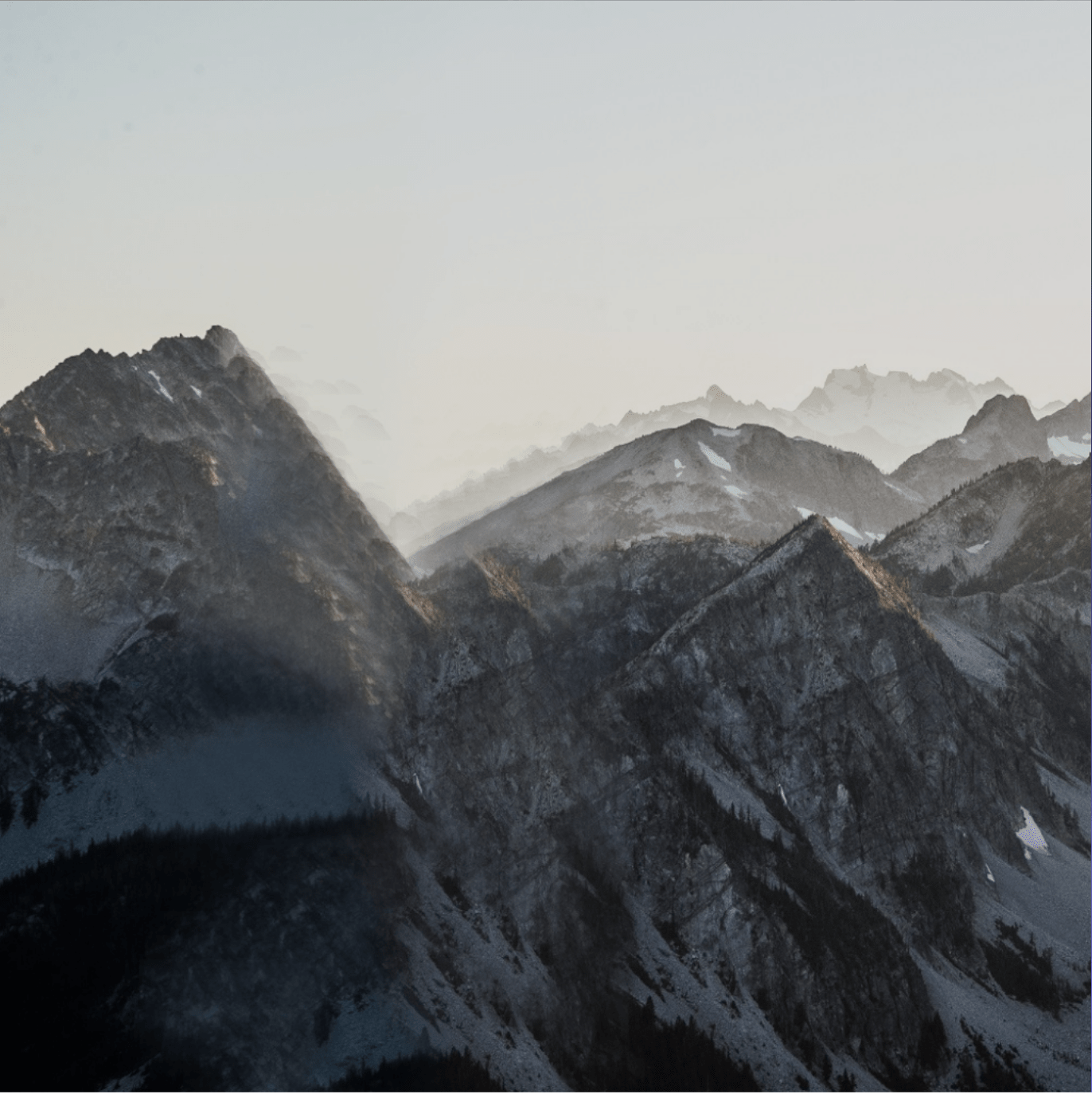}
        \caption{Restored result}
        \label{fig:onefoldresult_c}
    \end{subfigure}
    \caption{Unfolding and restoring result of a 1-folded image.}\label{fig:mountain_result}
\end{figure}


Figure \ref{fig:mountain_result} shows the result for the 1-folded example. The unfolding is trivial in a sense, but here we want to use it to illustrate the typical procedure of a folded-image restoration. The same procedure applies to all kinds of folds alike. In any case, we assume to have the partial boundary and singular set data (i.e., the folding edges and the boundaries) available from the folded images.  Our goal is to recover the folded image \ref{fig:onefold_a} by using the proposed unfolding technique and some well-established in-painting algorithms. At the beginning of our algorithm, an initialization $\Sigma_0$ is constructed by simply using the partial boundary and singular set data. By the proposed Algorithm 1, we can successfully reconstruct the folding map based on the reinforcement iterations, which are shown in \ref{fig:onefoldresult_b}. With the folding map, we can obtain the occluded region and carry out the in-painting procedure.  Notice that the unfolded mesh in \ref{fig:onefoldresult_b} highly resembles to the ground truth (see \ref{fig:onefoldresult_a}). With such unfolded domain, we can acurately approximate the masked region and apply the patch-matching based in-painting algorithm to recover the image. The corresponding parametrized folding surfaces are show in figure \ref{fig:onefold_b}. The image can then be mapped from the folded surface back to the domain complementary to the occluded region. Here, as is common the case, due to the size of masked region generated from the fold of the image, we choose the patch-matching based algorithm for image in-painting. In particular, we employ the algorithm\footnotemark proposed by Daisy et. al. \cite{daisy2013fast} in this example. The result is shown in \ref{fig:onefoldresult_c}. For comparison, the original image together with the overlapping mask (drawn as a half-transparent domain) is shown in Figure \ref{fig:ex1_mask}.

\footnotetext{The algorithm is available as a plugin for the open source GIMP2 software. The software is available at \url{https://www.gimp.org/}.}

To illustrate the adaptability of our proposed algorithm in some more complicated folded surface, we now consider correspondingly, a 2-, 3-folded examples, as shown in Figure \ref{fig:experiemnt234-fold}. There is a 2-folded painting ``The Sower" by Vincent van Gogh, and a  cusp-folded ``building" image.

Similarly, to approximate the original images with the given folded data, we first have to unfold these surfaces using the proposed algorithm. Unlike the trivial 1-fold example illustrated above (which may be simply get unfolded even without the use of Algorithm 1), the folding order also takes part in this unfolding problem since different orders of folding produce different folded images. When the folding number is large, obtaining the ordering of the folds from the given data is difficult. However, using the Algorithm 1, this folding order can be obtained implicitly. In other words, the unfolding procedure using Algorithm 1 is fully automatic. Figure \ref{fig:234fold-result} shows the unfolding and the corresponding in-painting results. The leftmost column shows the folded meshes corresponding to some unidentified rectangular meshes. With only partial boundary conditions and singular set data, unfolding these surfaces are highly ill-posed. But using our proposed algorithm, we successfully obtained the unfolding surfaces (the middle-left column). Notice that by regularizing the generalized Beltrami coefficient, Algorithm 1 converges to unfolded regular meshes, where unnatural curvy edges are not presented. With these unfolded meshes, we can recover the occluded regions due to the foldings (See the middle-right column) and therefore in-painting algorithms can be employed as usual. The overall recovered images are shown in the rightmost column of Figure \ref{fig:234fold-result}.

\begin{figure}
    \centering
    
    \begin{subfigure}[b]{0.23\textwidth}
        \includegraphics[width=\textwidth]{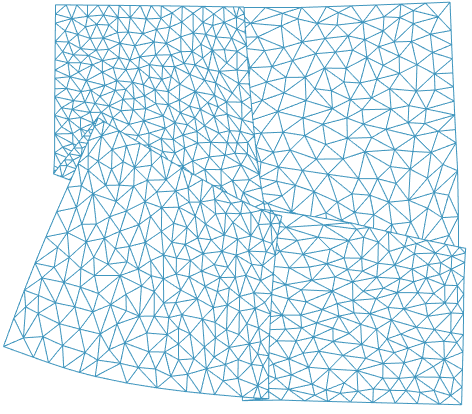}
    \end{subfigure}
    \begin{subfigure}[b]{0.22\textwidth}
        \includegraphics[width=\textwidth]{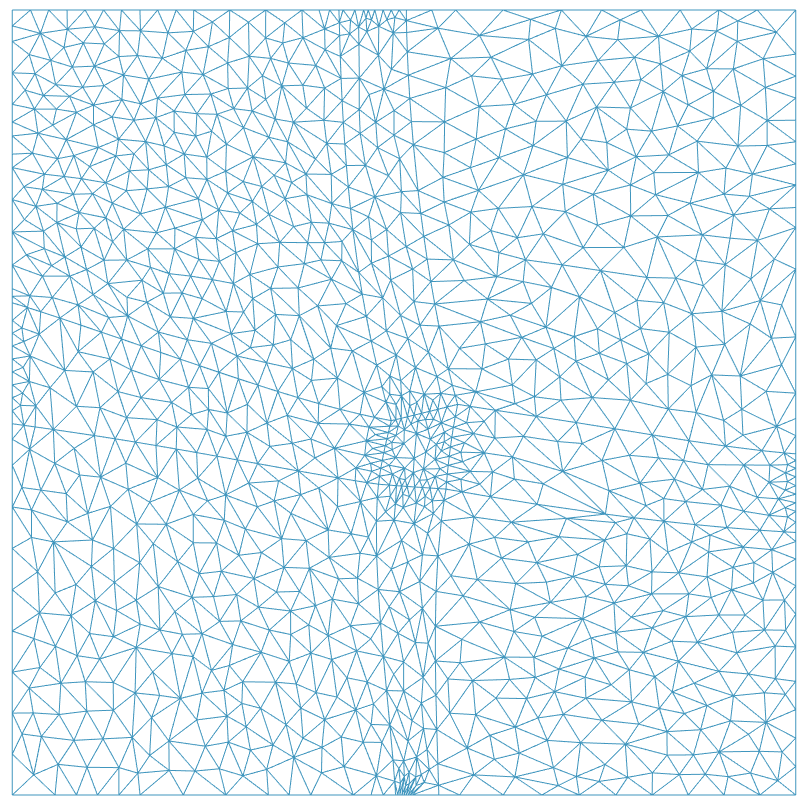}
    \end{subfigure}
    \begin{subfigure}[b]{0.22\textwidth}
        \includegraphics[width=\textwidth]{sower_sq_with_mask.png}
    \end{subfigure}
    \begin{subfigure}[b]{0.22\textwidth}
        \includegraphics[width=\textwidth]{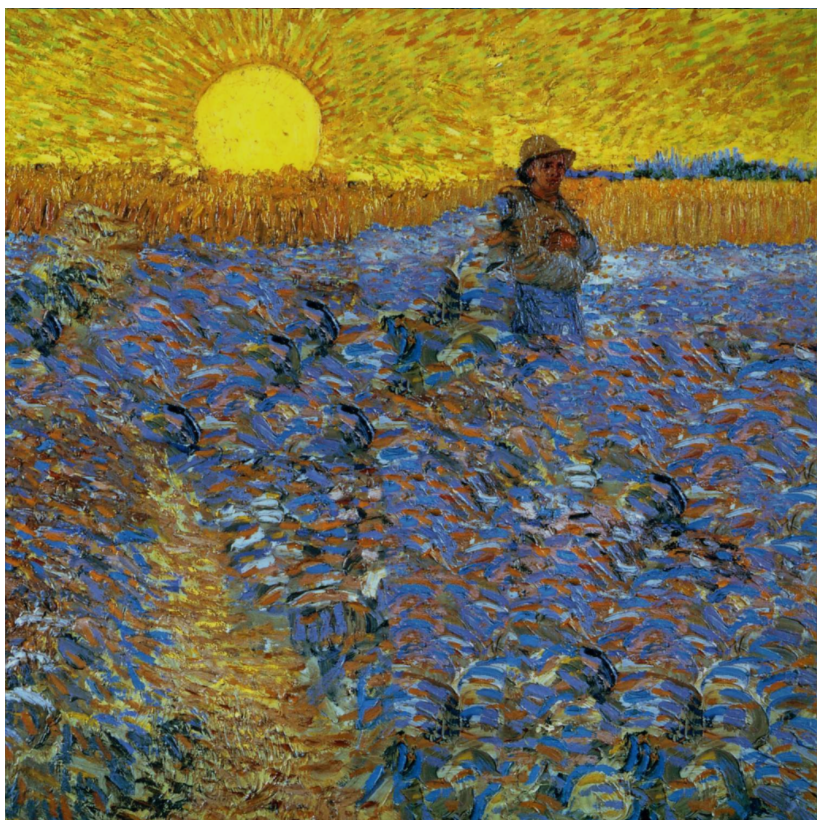}
    \end{subfigure}
    \begin{subfigure}[b]{0.22\textwidth}
        \includegraphics[width=\textwidth]{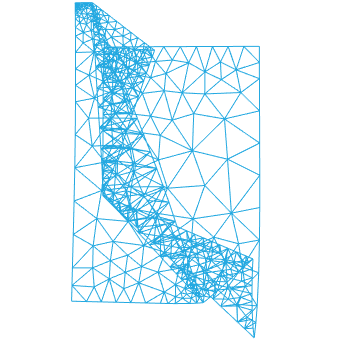}
    \end{subfigure}
    \begin{subfigure}[b]{0.21\textwidth}
        \includegraphics[width=\textwidth]{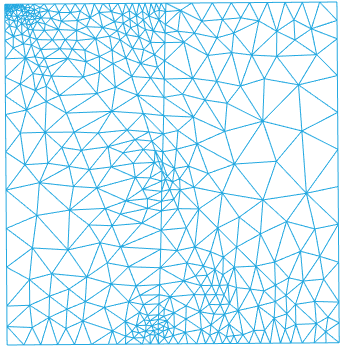}
    \end{subfigure}
    \begin{subfigure}[b]{0.22\textwidth}
        \includegraphics[width=\textwidth]{building_masked.png}
    \end{subfigure}
    \begin{subfigure}[b]{0.22\textwidth}
        \includegraphics[width=\textwidth]{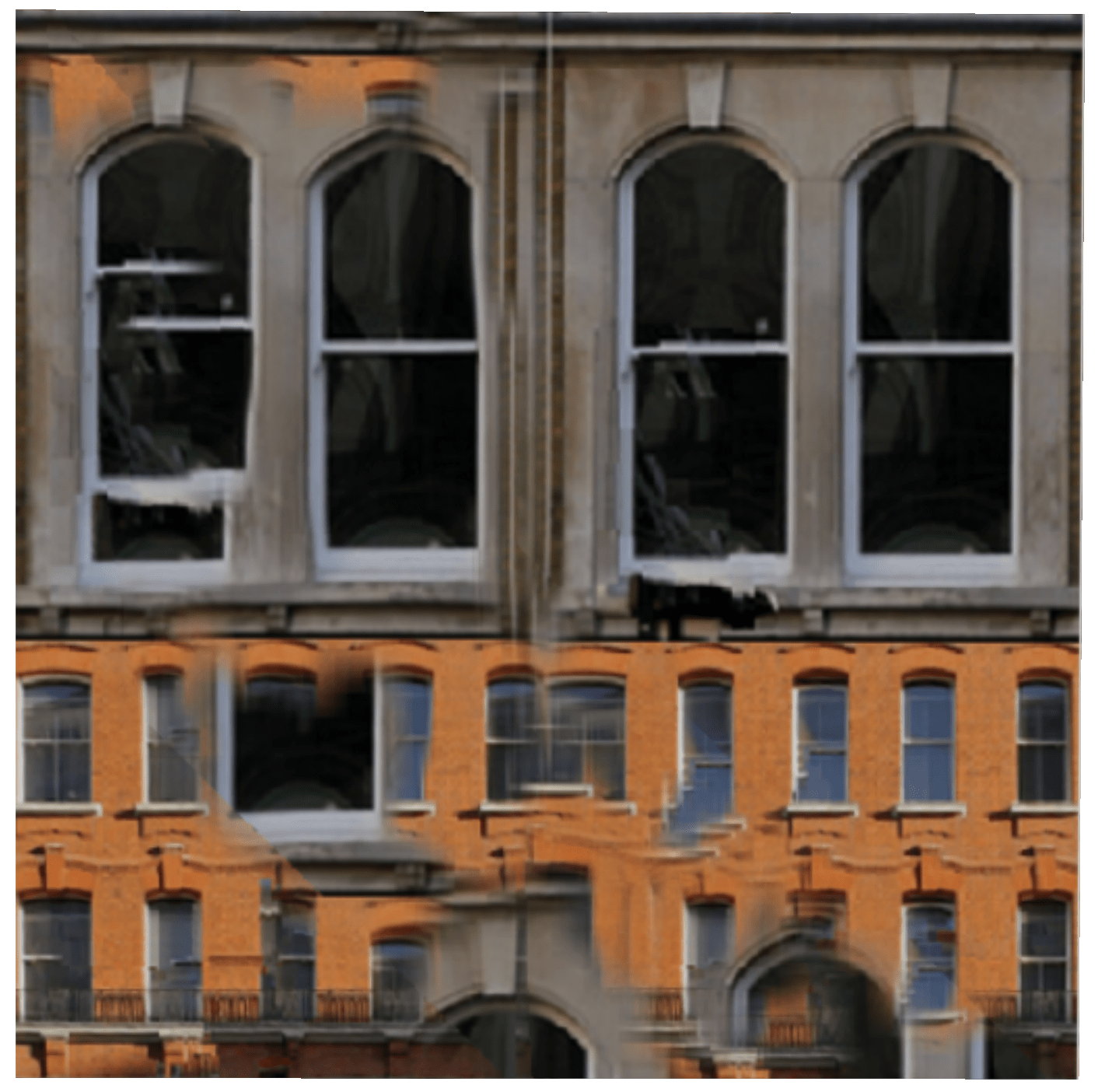}
    \end{subfigure}

    \caption{Unfolding and in-painting result for the 2-fold (first row) and cusp-fold (second row) examples. Leftmost column shows the folded meshes representing the domain of the images. The unfolded results are shown in the middle-left column and the recovered occluded domains are shown in the middle-right columns. The overall recovery results are shown in the rightmost column.}
    \label{fig:234fold-result}
\end{figure}

\section{Discussion and Conclusion}

We have proposed a novel way of studying and modeling the folding phenomena of surfaces using alternating Beltrami equations. The numerical scheme is proposed to discretize and solve the equation, by taking into account of the coupled nature of the two coordinate functions of the solution. The resulting method works for as few as fixing two point as constraints, and has a nice geometric interpretation. More importantly, it allows us to formulate and solve the inverse problem of inferring and parametrizing flat-foldable surfaces with observed partial data. We have proposed to use the ``reinforcement iteration" algorithm in order to solve the associated optimization problem, which has shown empirical convergence over various examples. Various applications are given, including fold sculpting, fold-like texture generation, generating and editing generalized Miura-ori patterns, as well as self-occlusion reasoning. Many more possible applications shall be explored in the future. At the same time, the understanding of non-rigid folding is still largely incomplete and many more interesting examples and applications are waiting to be discovered.

\section*{Acknowledgments}
The first author would like to thank Mr. Leung Liu Yusan and Dr. Emil Saucan for some useful help and discussions in the early stage of this work. The examples meshes are generated by the software Triangle \cite{shewchuk1996triangle}. This work is supported by HKRGC GRF (Project ID: 2130549).

\bibliographystyle{siamplain}
\bibliography{references}

\end{document}